\documentclass[11pt]{article}

\usepackage[bottom]{footmisc}
\usepackage{amssymb}
\usepackage{amsmath}
\usepackage{amsthm}
\usepackage{amsfonts,amsthm,amsmath,amssymb}

\usepackage[
    disableredefinitions,
    classfont=bold,
    full,
    langfont=caps,
    funcfont=roman
]{complexity}

\usepackage[mathscr]{euscript}
\usepackage[margin=2cm]{geometry}
\usepackage{bm}
\usepackage[
	backend=bibtex8,
	style=alphabetic,
	natbib=true,
	maxnames=10,
	maxitems=10,
	maxcitenames=3,
	maxalphanames=4,
	minalphanames=3,
	labelalpha=true,
]{biblatex}

\usepackage{dsfont}

\usepackage{hyperref}
\hypersetup{
    citecolor=red,
    colorlinks=true,
    linkcolor=blue,
    filecolor=magenta,      
    urlcolor=cyan,
}

\usepackage{thmtools}
\usepackage{thm-restate}
\usepackage[noabbrev,capitalise]{cleveref}
\usepackage{mathtools}

\declaretheorem[numberwithin=section]{thm}
\declaretheorem[sibling=thm,name=Theorem]{theorem}
\declaretheorem[sibling=thm,name=Lemma]{lemma}
\declaretheorem[sibling=thm,name=Corollary]{corollary}
\declaretheorem[sibling=thm,name=Claim]{claim}

\declaretheorem[sibling=thm,name=Fact]{fact}

\declaretheorem[sibling=thm,name=Definition]{definition}

\usepackage{soul}
\usepackage{xcolor}
\usepackage{wasysym}
\usepackage{comment}
\usepackage{marvosym}
\usepackage[shortlabels]{enumitem}
\usepackage{tikz}
\usetikzlibrary{calc}

\usepackage{braket}

\mathchardef\mhyphen="2D

\newcommand{\eps}{\varepsilon}
\newcommand{\inr}{\in_R}

\newcommand{\defeq}{\stackrel{def}{=}}
\newcommand{\defeqq}{\coloneqq}
\newcommand{\col}{\colon}

\newcommand{\qlq}{\quad\Longrightarrow\quad}

\mathchardef\mhyphen="2D

\newcommand{\wh}[1]{\widehat{#1}}
\newcommand{\wt}[1]{\widetilde{#1}}
\newcommand{\ol}[1]{\overline{#1}}

\newcommand{\cB}{\mathcal{B}}

\newcommand{\cD}{\mathcal{D}}
\newcommand{\cE}{\mathcal{E}}

\newcommand{\cI}{\mathcal{I}}

\newcommand{\cM}{\mathcal{M}}

\renewcommand{\cS}{\mathcal{S}}
\newcommand{\cT}{\mathcal{T}}
\newcommand{\cU}{\mathcal{U}}
\newcommand{\cV}{\mathcal{V}}

\renewcommand{\C}{\mathbb{C}}

\newcommand{\N}{\mathbb{N}}

\crefname{claim}{claim}{claims}
\Crefname{claim}{Claim}{Claims}
\crefname{algorithm}{algorithm}{algorithms}
\Crefname{algorithm}{Algorithm}{Algorithms}
\crefname{algocf}{algorithm}{algorithms}

\makeatletter
\if@cref@capitalise
\crefname{claim}{Claim}{Claims}
\else
\crefname{claim}{claim}{claims}
\fi
\makeatother

\newcommand{\norm}[1]{\left\| #1 \right\|}

\renewcommand{\set}[1]{{\left\{ #1 \right\}}}

\newcommand{\sett}[2]{\left\{ #1 \;\middle\vert\; #2 \right\}}

\newcommand{\abs}[1]{{\left\lvert{#1}\right\rvert}}

\newcommand{\supp}{\mathrm{Supp}}

\newcommand{\power}[1]{\{0,1\}^{#1}}
\newcommand{\pown}{\power{n}}
\newcommand{\powern}{\pown}

\DeclarePairedDelimiter{\parentheses}{\lparen}{\rparen}
\newcommand{\ps}[1]{\parentheses*{#1}}
\newcommand{\pss}[1]{(#1)}

\newcommand{\ketbra}[2]{\ket{#1}\bra{#2}}
\newcommand{\ketbraa}[1]{\ketbra{#1}{#1}}

\renewcommand{\braket}[2]{\left\langle{#1} \middle|{#2}\right\rangle}

\newclass{\TISP}{TISP}
\newclass{\RTISP}{RTISP}

\newclass{\RTISPs}{R^*TISP}
\newclass{\BPTISP}{BPTISP}
\newclass{\BPTISPs}{BP^*TISP}
\newclass{\prBPTISPs}{prBP^*TISP}

\newclass{\BPSs}{BP^*SPACE}
\newclass{\RSs}{R^*SPACE}

\newclass{\prBPSs}{prBP^*SPACE}
\newclass{\prRSs}{prR^*SPACE}

\newclass{\searchBPP}{SearchBPP}
\newclass{\searchP}{SearchP}

\renewclass{\promiseBPP}{prBPP}
\renewclass{\promiseRP}{prRP}
\newclass{\prRTISPs}{pr\RTISPs}

\newclass{\prZPL}{prZPL}
\newclass{\prBPL}{prBPL}
\newclass{\prRL}{prRL}
\newclass{\prL}{prL}

\newclass{\searchNL}{SearchNL}
\newclass{\searchBPL}{SearchBPL}
\newclass{\searchBPLs}{SearchBP^*L}
\newclass{\searchZPLs}{SearchZP^*L}
\newclass{\searchRLs}{SearchR^*L}
\newclass{\searchRL}{SearchRL}
\newclass{\searchL}{SearchL}

\newclass{\RLs}{R^{*}L}
\newclass{\BPLs}{BP^{*}L}
\newclass{\ZPLs}{ZP^{*}L}

\newclass{\BPPLs}{BP^*PL}
\newclass{\RPLs}{R^*PL}
\newclass{\prBPPLs}{prBP^*PL}
\newclass{\prRPLs}{prR^*PL}

\newclass{\prBPLs}{pr\BPLs}
\newclass{\prRLs}{pr\RLs}

\newclass{\searchBPSACE}{Search\BPSPACE}
\newclass{\searchDSPACE}{Search\DSPACE}

\newclass{\almost}{almost\mhyphen}
\newclass{\almostLw}{almost\textsubscript{w}\mhyphen L}
\newclass{\almostLs}{almost\textsubscript{s}\mhyphen L}
\newclass{\almostL}{\almost L}

\newclass{\recursiveset}{R}
\newclass{\MAL}{MAL}
\renewcommand{\L}{\mathbf{L}}
\renewcommand{\P}{\mathbf{P}}
\usepackage{booktabs}
\usepackage{multirow}
\usepackage{arydshln}
\usepackage{braket}
\usepackage{yhmath}
\usepackage{bbm}
\usepackage{float}
\usepackage[normalem]{ulem}

\newcommand{\stabgroup}{\mathrm{Stab}}

\newcommand{\Pauli}{\mathrm{Pauli}}

\DeclareMathOperator{\tr}{tr}

\newcommand{\rhoinit}{\rho}
\newcommand{\rhopre}{\rho^{xy} }
\newcommand{\rhopost}{\rho^{xyb}}

\newcommand{\genresource}{\mathsf{Gen}_{\mathrm{res}}}
\newcommand{\genalice}{{\mathsf{Gen}^{\mathrm{pre}}_{\mathtt{A}}}}

\newcommand{\genbob}{{\mathsf{Gen}^{\mathrm{pre}}_{\mathtt{B}}}}
\newcommand{\genalicepost}{\mathsf{Gen}^{\mathrm{post}}_{\mathtt{A}}}
\newcommand{\genbobpost}{\mathsf{Gen}^{\mathrm{post}}_{\mathtt{B}}}

\newcommand{\Qb}{{\ol {Q}}}
\newcommand{\rank}{\operatorname{rank}}

\newclass{\BQSPACE}{BQSPACE}
\newclass{\prBQTIME}{prBQTIME}
\newclass{\prBPTIME}{prBPTIME}

\NewDocumentEnvironment{spliteq}{b}{%
  \begin{equation}
    \begin{split}
      #1
    \end{split}
  \end{equation}
}{%
  \ignorespacesafterend
}

\begin{document}

\author{
    Oren Renard\thanks{Cornell University. Email: \texttt{oren@cs.cornell.edu}.}
    \and
    Nicholas Spooner\thanks{Cornell University. Email: \texttt{nspooner@cornell.edu}.}
}

\title{Computational Bounds for $f$-Routing}
\maketitle

\begin{abstract}
The $f$-routing protocol is a leading candidate for quantum position verification \cite{KentMunroSpiller2011}, but security guarantees for explicit functions remain limited.
We prove unconditional resource lower bounds for uniform attackers; our new techniques bypass communication-complexity bounds central to previous works, which are inherently at most linear in the input length.

We show that, for input length $n$ and sufficiently small constant $\eps>0$, a uniformly time $\cT$ strategy, with success probability $\ge 1-\eps$ on every input, implies the following computational bounds on $f$:
\begin{enumerate}
    \item If the strategies are arbitrary quantum channels, then $f\in \QSZK(\poly(\cT))$, where $\QSZK(\cT)$ is the class of languages having quantum statistical zero knowledge proofs in which the verifier runs in time $\cT$ (and the simulator in time $\poly(\cT)$).

    \item If the strategies are Clifford+T circuits using at most $t$ magic gates, then $f\in \DTIME(\poly(\cT 2^t))$.

    \item If the strategies are explicit Pauli-sparse unitaries (acting on at most $\cT$ qubits) that have at most $s$ nonzero Pauli coefficients, then $f\in \DTIME(\poly(\cT s))$.
\end{enumerate}
Time and space hierarchies then yield explicit functions secure against super-linear and even all-polynomial time $\mathcal{T}$ bounded adversaries under our computational restrictions. These bounds exceed the bounded-storage qubit bound $q\leq\log n$ of \cite{BluhmChristandlSpeelman2022} for the inner product function $f=\mathrm{IP}$, at the cost of restricting adversarial computation and increasing honest evaluation complexity.
\end{abstract}

\tableofcontents

\section{Introduction}\label{sec:introduction}
Quantum position verification is an exciting direction in Quantum Information theory and Cryptography that aims to turn a device's \emph{physical location} into a verifiable credential, enabling applications such as position-based authentication and key exchange \cite{BuhrmanEtAl2014Position}. Its central challenge is to certify a prover's location even when several devices elsewhere cooperate to impersonate it (see \cite{AlnawakhthaMiller2022} for a survey).

One idea for verifying a prover's location is to give it a distributed challenge, whose parts arrive from different locations, and time its response. Assuming we can estimate the communication and computation times required by all participants, a deadline for combining these parts and responding constrains the distance the messages could have traveled. This follows the principle of \emph{distance bounding} introduced by \citet{BrandsChaum1994}.

To make this concrete, we focus on the simplest setting, one-dimensional space, which already captures the essential challenge of position verification. Consider a prover $P$ claiming to be halfway between two verifiers $V_0$ and $V_1$ on a line (see \cref{fig:intro-spacetime}(a)).
The verifiers publicly announce a function $f\colon\powern\times\powern\to\power{}$ that specifies the prover's response to the distributed challenge.
Then, the verifiers send independent strings $x,y\in\powern$, timed to arrive at $P$ simultaneously. The prover must compute $b=f(x,y)$ and return $b$ to both verifiers within the prescribed time.

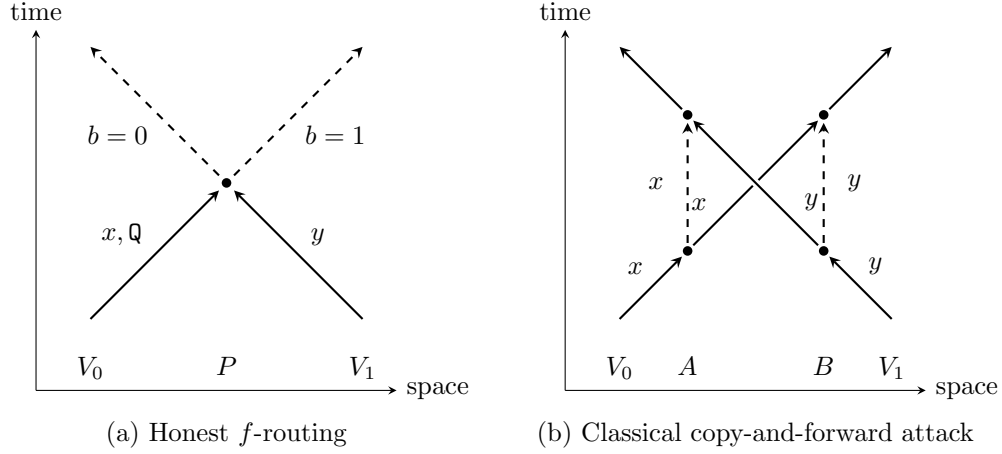
\begin{figure}[!ht]
\centering
\begin{tikzpicture}[x=0.9cm,y=0.9cm,>=stealth,font=\small]
  \begin{scope}
    \node at (2,-1.3) {(a) Honest $f$-routing};
    \foreach \x/\name in {0/V_0,2/P,4/V_1} {
      \node[below] at (\x,0) {$\name$};
    }
    \draw[->] (-0.8,-0.65) -- (-0.8,4.65) node[above] {time};
    \draw[->] (-0.8,-0.65) -- (4.5,-0.65) node[right] {space};
    \draw[->,thick,shorten >=4pt] (0,0.4) -- (2,2.4) node[pos=0.46,above left] {$x,\mathtt{Q}$};
    \draw[->,thick,shorten >=4pt] (4,0.4) -- (2,2.4) node[pos=0.46,above right] {$y$};
    \fill (2,2.4) circle (1.8pt);
    \draw[->,thick,dashed,shorten <=4pt] (2,2.4) -- (0,4.4) node[midway,below left] {$b=0$};
    \draw[->,thick,dashed,shorten <=4pt] (2,2.4) -- (4,4.4) node[midway,below right] {$b=1$};
  \end{scope}
  \begin{scope}[xshift=7cm]
    \node at (2,-1.3) {(b) Classical copy-and-forward attack};
    \foreach \x/\name in {0/V_0,1/A,3/B,4/V_1} {
      \node[below] at (\x,0) {$\name$};
    }
    \draw[->] (-0.8,-0.65) -- (-0.8,4.65) node[above] {time};
    \draw[->] (-0.8,-0.65) -- (4.5,-0.65) node[right] {space};
    \draw[->,thick,shorten >=3pt] (0,0.4) -- (1,1.4) node[midway,above left] {$x$};
    \draw[->,thick,shorten >=3pt] (4,0.4) -- (3,1.4) node[midway,above right] {$y$};
    \draw[->,thick,dashed,shorten <=3pt,shorten >=3pt,preaction={draw,white,solid,line width=3pt}] (1,1.4) -- (1,3.4);
    \draw[->,thick,dashed,shorten <=3pt,shorten >=3pt,preaction={draw,white,solid,line width=3pt}] (3,1.4) -- (3,3.4);
    \node[anchor=east] at (0.8,2.4) {$x$};
    \node[anchor=west] at (3.2,2.4) {$y$};
    \draw[->,thick,shorten <=3pt,shorten >=3pt] (1,1.4) -- (3,3.4) node[pos=0.22,above left] {$x$};
    \draw[->,thick,shorten <=3pt,shorten >=3pt,preaction={draw,white,line width=3pt}] (3,1.4) -- (1,3.4) node[pos=0.22,above right] {$y$};
    \draw[->,thick,shorten <=3pt] (1,3.4) -- (0,4.4);
    \draw[->,thick,shorten <=3pt] (3,3.4) -- (4,4.4);
    \foreach \x/\y in {1/1.4,3/1.4,1/3.4,3/3.4} {\fill (\x,\y) circle (1.8pt);}
  \end{scope}
\end{tikzpicture}
\caption{Spacetime diagrams with instantaneous local computation and communication at the speed of light. In (a), $P$ routes the quantum challenge along the dashed path selected by $b=f(x,y)$. In (b), $A$ and $B$ retain and exchange classical copies, allowing both to respond on time.}
\label{fig:intro-spacetime}
\end{figure}

As observed in the seminal work of \citet{ChandranEtAl2009}, classical protocols allow two colluding devices on opposite sides of the claimed position to keep copies of the messages they intercept and exchange them. Both then learn $f(x,y)$ in time to provide the required response to their nearby verifier. This ``copying'' attack, illustrated in \cref{fig:intro-spacetime}(b), underlies the impossibility of classical position verification without additional assumptions.

$f$-routing protocols \cite{KentMunroSpiller2011,BuhrmanEtAl2013GardenHose} build on the above idea by using $f$ to specify the destination of a quantum challenge. Verifier $V_0$ sends a quantum state, unknown to the prover, alongside $x$, and the prover must forward it unchanged to $V_{f(x,y)}$. While the no-cloning principle prevents the classical copying attack, colluding adversaries sharing a large enough entangled state can nonetheless implement the routing task using Non-Local Quantum Computation \cite{Vaidman2003,BeigiKonig2011,BuhrmanEtAl2014Position} (see also the recent monograph by \citet{May2026EntanglementCost}). The entanglement cost of the $f$-routing task for general $f$ remains a wide open question, despite extensive work over the past decade \cite{TomamichelFehrKaniewskiWehner2013,ChiuEtAl2014,KlauckPodder2014,RibeiroGrosshans2015,Speelman2016,ArunachalamPodder2021,BluhmChristandlSpeelman2022,CreeMay2023,AllerstorferEtAl2024,AsadiCleveCulfMay2025,AsadiCulfMay2025}.

The challenge is to choose $f$ that an honest prover can compute efficiently once both inputs arrive, but whose routing task is costly for spatially separated adversaries. We study this cost in quantum storage and in the resources needed to prepare the shared state and perform local operations.

Throughout this work, we consider this arrangement of two colluding adversaries on opposite sides of the claimed position, whom we call Alice ($A$) and Bob ($B$), located near $V_0$ and $V_1$, respectively. Informally, they intercept their respective challenges, perform local computations, and simultaneously exchange messages before responding to their nearby verifiers.
Establishing establish security of $f$-routing against such adversaries under various restrictions is the main goal of this work.
We use the standard worst-case correctness requirement for implementing the $f$-routing task, requiring a single strategy that succeeds in the protocol with probability $\ge1-\eps$ on every input pair $(x,y)$. Proving worst-case lower bounds is already challenging  \cite{BuhrmanEtAl2013GardenHose,AllerstorferEtAl2024,AsadiEtAl2025CDQS}.

\paragraph{Known Results.}

Teleportation-based attacks can implement $f$-routing for any function on $n$-bit inputs using $2^{O(n)}$ shared EPR pairs (see \cite{Vaidman2003,BeigiKonig2011,BuhrmanEtAl2014Position,KentMunroSpiller2011,BuhrmanEtAl2013GardenHose} and references therein).  A recent connection to Conditional Disclosure of Secrets yields a perfect attack for every $f$ using $2^{O(\sqrt{n\log n})}$ qubits of shared resource and quantum communication (see \cite{AllerstorferEtAl2024,LiuVaikuntanathanWee2017}).

On the positive side, \citet{Unruh2014} proved security of an $f$-measure-based position verification protocol in the random oracle model against adversaries making polynomially many oracle queries, even with unbounded pre-shared entanglement.
\citet{ColissonPalaisEscolaFarrasSpeelman2025} showed that $m$ parallel repetitions of their basic $f$-routing protocol achieve exponentially small soundness $2^{-\Omega(m)}$ against adversaries sharing at most $\alpha m$ entangled qubits, for a sufficiently small $\alpha >0$; the drawback of this protocol is that its quantum resources scale with $m$.

Further security results and lower bounds under various restrictions on the adversaries were studied as well (e.g., Garden Hose attacks \cite{BuhrmanEtAl2013GardenHose}, classical communication \cite{RibeiroGrosshans2015}, span programs \cite{CreeMay2023}, and much more \cite{KlauckPodder2014,RibeiroEtAl2018,May2022Complexity,AsadiCleveCulfMay2025,AsadiCulfMay2025}).

For the most general adversarial models, two main works give resource lower bounds \cite{AsadiCleveCulfMay2025,BluhmChristandlSpeelman2022}. Here, $q$ denotes an upper bound on the number of qubits held by each adversary. The former \cite{AsadiCleveCulfMay2025} proved a lower bound of $\Omega(\tfrac{n}{\log q})$ on the number of quantum gates and measurements. The latter, due to \citet{BluhmChristandlSpeelman2022}, studied the bounded quantum-storage model.
They established information-theoretic security $f$-routing when each of the two adversaries initially holds at most $q$ qubits in the purified model\footnote{The qubit bound includes all quantum workspace, including auxiliary registers needed to represent measurements as unitary operations (see \cite{BluhmChristandlSpeelman2022}).}, while each adversary is computationally unbounded. Their explicit construction sets $f$ to be the inner-product function $\mathrm{IP}(x,y)=\sum_{i=1}^n x_i y_i\bmod 2$.

\begin{theorem}[{\cite{BluhmChristandlSpeelman2022}}]\label{thm:intro-bcs}
    The $\mathrm{IP}$-routing protocol is secure against adversaries whose purified strategy initially uses at most $q\leq\frac12\log_2 n-c$ qubits per adversary, for a universal constant $c$.
\end{theorem}

We emphasize that the storage restriction in \cref{thm:intro-bcs} is essential for the inner-product function. Indeed, $\mathrm{IP}$-routing admits a perfect attack using $O(n)$ shared EPR pairs, since its garden-hose complexity satisfies $\mathrm{GH}(\mathrm{IP})=\Theta(n)$ \cite{BuhrmanEtAl2013GardenHose,KlauckPodder2014}.

Additionally, \cite{BluhmChristandlSpeelman2022} also showed that almost every Boolean function $f$ yields security against adversaries holding at most $n$ qubits. However, almost all such functions require exponentially large circuits to evaluate, making them impractical for an efficient honest prover. It remains open to describe an explicit function for which $f$-routing requires superlogarithmic entanglement.

\subsection{Main Results}\label{sec:main-results}

In this work we study the relationship between the computational complexity of $f$ and the computational resources required for an attack against $f$-routing.

Our main contribution is a security reduction that turns a successful $f$-routing strategy into an algorithm for computing $f$. At a high level, our idea is to relate $f(x,y)$ to the quantum information produced at an intermediate stage of the strategy, before the adversaries can combine both inputs to evaluate $f(x,y)$ directly. We use their circuit descriptions to reproduce this information and reduce computing $f(x,y)$ to an instance of the (complement of) Quantum State Distinguishability (QSD) problem. The circuits in the resulting QSD instance are closely related to the adversaries' initialization and local circuits.

To begin with, let us introduce the adversarial model.
A strategy for $f$-routing consists of three parts: an initialization channel, which prepares the initial entangled state of the adversaries; pre-communication channels for Alice and Bob, which apply an operation to the local systems based on the corresponding local inputs; and post-communication or ``decoding'' channels which extract the routed qubit.
Henceforth we suppose the strategies are uniform in the following sense.

\begin{definition}\label{def:uniform-strategies-informal}
    Let $\cT \col \N \to \N$ be nondecreasing and time-constructible, with $\cT(n)\geq n$.
    A family of routing strategies is \emph{$\cT$-uniform} if it is described by three classical Turing machines $(\genresource,\genalice,\genbob)$
    such that, for every $n\in\N$:
    \begin{enumerate}
        \item $\genresource(1^n)$ prints a description of a quantum channel preparing the adversaries' initial entangled resource state.
        \item For every $x,y \in \powern$, the machines $\genalice(1^n, x)$ and $\genbob(1^n, y)$ print descriptions of Alice's and Bob's quantum channels implementing their local algorithms on inputs $x$ and $y$, respectively. These algorithms process their local resources and determine which message (quantum or classical) to send to the other adversary.
    \end{enumerate}
    Moreover, all the above machines run in time at most $\cT(n)$.
\end{definition}
In particular, the circuit printed by each generator has size at most $\cT(n)$, and acts on that many qubits. For specific parameters, the reader should keep in mind any fixed polynomial bound $\cT(n)=\poly(n)$, or even super-polynomial bounds. We discuss such instantiations in \cref{sec:explicit-security}.

Our resource-bounded model is motivated by the wide gap between known attacks and security guarantees for $f$-routing. For any $f$, there is a perfect attack using $2^{O(\sqrt{n\log n})}$ qubits of shared resource and quantum communication \cite{AllerstorferEtAl2024}. At the other extreme, position verification can be information-theoretically secure without preshared entanglement, even against computationally unbounded adversaries (see \cite{BuhrmanEtAl2014Position,ColissonPalaisEscolaFarrasSpeelman2025}). We study strategies between these extremes, seeking security against larger attacks while keeping honest evaluation of $f$ efficient.

Our model charges for preparing the initial shared state and for generating and carrying out the adversaries' local operations before communication, including any classical or quantum processing of $x$ and $y$ used in those operations. Computation unrelated to the initialization or these input-dependent operations is free.
We emphasize that, at a later stage of the protocol, we permit Alice and Bob to process information received from one another and even to know $f(x,y)$. By this stage, the honest prover would already have had to compute $f(x,y)$ to route correctly. We therefore grant the adversaries this value for free (at this space-time point), regardless of the computational cost of obtaining it.\footnote{Alice and Bob simultaneously send one message each, which may contain quantum systems and copies of their classical inputs. They then apply local recovery channels to their retained and received systems, as formalized in \cref{sec:setup}.}

% Put differently, what makes this model nontrivial is that we do not require the post-communication phase, which explicitly depends on both xand y, to be computationally bounded; the class of adversaries we capture can (eventually) compute f(x,y).

Our full results are generalized for arbitrary qudit dimension $d$, but we focus on $d=2$ in this overview.

Our main result is a reduction from evaluating $f$ in a given input $(x,y)$, to an instance of the (complement of) the Quantum State Distinguishability (QSD) problem \cite{Watrous02}.
\citet{Watrous02,Watrous09} showed that QSD is complete for $\QSZK$.

To track the resources of this reduction, we introduce the time-bounded class $\QSZK(\cT)$. Its languages admit quantum statistical zero-knowledge proofs with honest verifier running time at most $\cT(n)$ and simulation time polynomial in $\cT(n)$ and the dishonest verifier's running time (see \cref{def:qszk-time}).

\begin{theorem}\label{thm:main-qszk}
    There exist universal constants $\eps_0 >0$ and $c\geq1$ such that, if a $\cT$-uniform $f$-routing strategy is correct with probability $\ge (1-\eps)$ on all inputs $x,y \in \powern$ for some $\eps\leq\eps_0$, then
    \[
    \ensuremath{f\in\QSZK\ps{\cT^c}.}
    \]
\end{theorem}
Here $f\in\QSZK(\cT)$ means that the language $L_f = \sett{x \circ y}{f(x,y) = 1}$ belongs to $L_f \in \QSZK(\cT)$.

Our uniform reduction yields security of $f$-routing for \emph{explicit} $f$. Since $\QSZK(\cT)\subseteq\DSPACE(\cT^a)$ for a universal constant $a\geq1$ (following from $\QIP = \PSPACE$ \cite{JainJiUpadhyayWatrous11} and a standard padding argument), the deterministic space hierarchy gives the following unconditional security guarantee.

\begin{corollary}\label{thm:main-explicit-security}
    There exist universal constants $\eps_0 >0$ and $c'\geq1$ such that there is an explicit family of total Boolean functions $f\in\DSPACE(\cT^{c'})$ for which no $\cT$-uniform bounded strategy is $(1-\eps_0)$-correct on all inputs.
\end{corollary}
The main advantage of \cref{thm:main-explicit-security} is that it tolerates arbitrary quantum channels (within the stated resource bounds). We illustrate the resulting security bounds and compare them with prior work in \cref{sec:explicit-security}.
In particular, the above yields unconditional security against adversaries holding arbitrarily polynomial qudits $q= \poly(n)$, whenever $\cT(n) = \poly(n)$, from an explicit function $f$.
We reach this regime by allowing more time for honest evaluation while bounding its space usage. This relaxation is, in fact, \emph{necessary}, as security against all polynomial-time adversaries requires superpolynomial honest evaluation time, unless one separates $\L$ from $\P$ (see \cite{BuhrmanEtAl2013GardenHose}).

We next obtain tighter computational bounds on $f$ by further restricting the adversaries' resources.
We first consider $(\cT,t)$-uniform Clifford+$T$ strategies, which are $\cT$-uniform and use at most $t$ magic gates in each circuit.
For these strategies, we obtain an improved, deterministic algorithm for computing $f$, with running time polynomial in the input length, number of qubits, and circuit size, and exponential only in the number of magic gates.
\begin{theorem}\label{thm:main-clifford-security}
    There exist universal constants $\eps_0>0$ and $c\geq1$ such that if a $(\cT,t)$-uniform Clifford+$T$ $f$-routing strategy is $(1-\eps)$-correct on every input at every input length for some $0\leq\eps\leq\eps_0$, then
    \[
    \ensuremath{f\in\DTIME\ps{(\cT 2^t)^{c}}.}
    \]
\end{theorem}

Applying the deterministic time hierarchy again yields explicit functions secure against such strategies.

\begin{corollary}\label{cor:main-clifford-security}
    There exist universal constants $\eps_0 >0$ and $c\geq1$ such that there there is an explicit family of total Boolean functions $f\in \DTIME((\cT 2^t)^{c+1})$ for which no $(\cT,t)$-uniform Clifford+$T$ strategy is correct with probability $\ge (1-\eps_0)$ on all inputs.
\end{corollary}
Our reduction goes beyond the linear gate bound of \citet{AsadiCleveCulfMay2025} for $(\cT,t)$-uniform Clifford+$T$ strategies. It tolerates any fixed polynomial circuit-size bound, or even all polynomial bounds, at the cost of choosing $f$ from a suitably higher level of the deterministic time hierarchy and increasing the honest evaluation time accordingly.

Secondly, we consider $(\cT,s)$-uniform Pauli-sparse strategies, that restrict the adversaries to be implemented by unitaries on at most $q\leq\cT$ qubits, each supported on at most $s$ coefficients in the Pauli basis. This sparsity allows us to compute $f$ in time polynomial in $\cT$ and $s$.
We consider two cases: explicit strategies, in which the Pauli terms and their coefficients can be printed in deterministic time $\mathcal{T}$, and non-explicit strategies, in which the generators print circuits implementing the unitaries without specifying their Pauli decompositions (which can be viewed as oracle access to the unitaries, with query time $O(\mathcal{T})$). In the former case we obtain a deterministic classical algorithm; in the latter, we obtain a bounded-error quantum algorithm.

\begin{theorem}\label{thm:main-pauli-security}
    There exist universal constants $\eps_0>0$ and $c\geq1$ such that the following holds. Suppose a $(\cT,s)$-uniform Pauli-sparse $f$-routing strategy is correct with probability $\ge(1-\eps)$ on every input at every input length for some $\eps\leq\eps_0$.
    \begin{enumerate}
        \item If the strategy is explicit, then $f\in\DTIME\ps{(\cT s)^c}$.

        \item If the strategy is non-explicit, then $f\in\BQTIME\ps{(\cT s)^c}$.
    \end{enumerate}
\end{theorem}
The deterministic time hierarchy therefore gives explicit functions secure against explicit Pauli-sparse strategies. For non-explicit strategies, an analogous uniform bounded-error quantum time hierarchy with polynomial overhead is unknown (see discussion in \cref{sec:discussion}).

\begin{corollary}\label{cor:main-pauli-security}
    There exist universal constants $\eps_0 >0$ and $c\geq1$ such that there is an explicit family of total Boolean functions $f\in\DTIME((\cT s)^{c+1})$ for which no explicit $(\cT,s)$-uniform Pauli-sparse strategy is correct with probability $\ge(1-\eps_0)$ on all inputs.
\end{corollary}

The main theorem driving our reductions (\cref{thm:main-qszk,thm:main-clifford-security,thm:main-pauli-security}) is a new information-theoretic result for $f$-routing.
To state this result, let us first formalize the protocol and the quantum state we examine.

In the standard entanglement-based formulation of $f$-routing, imagine that the verifiers prepare an EPR pair on $\tt{\Qb Q}$. They keep $\tt{\Qb}$ as a reference and send $\tt Q$ as the challenge qubit. Thus, the prover must route the challenge $\tt Q$ to $V_{f(x,y)}$ while preserving its entanglement with $\tt{\Qb}$. We examine the \emph{mid-protocol state} $\rhopre_{\tt{\Qb AB}}$, where $\tt A,\tt B$  are Alice's and Bob's systems after they each process their own input locally and send a message to the other based on that processing. Crucially, we can prepare this state without knowing $f(x,y)$, at a cost comparable to their  initialization and local operations costs.
See \cref{sec:setup} for the formal definition and \cref{fig:main-nlqc} for an illustration of these subsystems and the adversaries' shared entanglement.

Our theorem shows that successful routing forces the reference subsystem $\tt{\Qb}$ to be nearly uncorrelated with the adversary who should not receive the challenge. In contrast, it must remain correlated with the intended recipient. We quantify this distinction through the fidelity of the joint state to the product of its marginals.

\begin{theorem}\label{thm:main-fidelity-gap}
    Let a routing strategy be $(1-\eps)$-correct on input $(x,y)$, and let $\rho^{xy}_{\ol{\mathtt Q}\mathtt{AB}}$ be its mid-protocol state. Define
    \[
    F_0(x,y)\defeqq F\ps{\rho^{xy}_{\ol{\mathtt Q}\mathtt A}\; ,\;\rho^{xy}_{\ol{\mathtt Q}}\otimes\rho^{xy}_{\mathtt A}}.
    \]
    Then
    \[
    \begin{aligned}
        f(x,y)=1 &\quad\Longrightarrow\quad F_0(x,y)\geq1-4\eps,\\
        f(x,y)=0 &\quad\Longrightarrow\quad F_0(x,y)\leq 1/2 +\sqrt{2\eps}.
    \end{aligned}
    \]
\end{theorem}
Intuitively, our bounds quantify the correlations between each adversary and the reference system at the midpoint of the protocol. As expected, the intended recipient of the routed qubit remains entangled with the reference, while the other adversary is nearly uncorrelated with it.

Furthermore, our result extends the product-state characterization of \cite{AsadiCulfMay2025}, which assumes 1-side perfect correctness $\eps=0$, and small error for the other. Our theorem allows nonzero routing error $\eps >0$ for both output values and gives quantitative completeness and soundness bounds for the fidelity test. Our proof also uses different tools, replacing the Information argument of \cite{AsadiCulfMay2025} with direct fidelity inequalities that accommodate nonzero error on both output values.

Therefore, the above implies that for sufficiently small $\eps$, approximating $F_0(x,y)$ to a small constant additive error determines $f(x,y)$.

\subsection{Explicit Security}\label{sec:explicit-security}
We now instantiate \cref{thm:main-explicit-security,cor:main-pauli-security,cor:main-clifford-security} to compare the adversaries' resources with the cost of honest evaluation. Our objective is to relate the adversaries uniform time bound $\cT$ to the function's input length $2n$. The time bound also bounds the number of qubits used by the adversaries (i.e., $q(n)\leq\cT(n)$). For the restricted models, we also track the Pauli sparsity $s$, and Magic gate count $t$. As noted above, we assume the generation time of the strategies is bounded by $\cT$. \Cref{tab:explicit-security} summarizes the resulting parameter choices and honest evaluation costs.

\begin{table}[htbp]
    \centering
    \small
    \setlength{\tabcolsep}{4pt}
    \renewcommand{\arraystretch}{1.2}
    \begin{tabular}{|l@{\hspace{8pt}}|l|l|}
        \hline
        \textbf{Adversary model} & \textbf{Parameter regime} & \textbf{Honest evaluation cost} \\
        \hline
        \multirow{2}{*}{Time $\cT$ channels} & $\cT=\poly(n)$ & Polynomial space \\
        \cdashline{2-3}
        & $\cT= n^{\omega(1)}$ & Quasi-olynomial space \\
        \hline

        \multirow{2}{*}{\shortstack[l]{Time $\cT$ \\Explicit $s$-Pauli-sparse}} & $\cT,s=\poly(n)$ & Polynomial time \\
        \cdashline{2-3}
        & $\cT,s=2^{\poly(\log n)}$ & Quasipolynomial time \\
        % \hline
        %
        % \multirow{2}{*}{Non-Explicit $(q,s)$-Pauli-sparse} & $q,s=\poly(n)$ & Polynomial quantum time \\
        % \cdashline{2-3}
        %  & $q=\poly(n)$, $s=2^{\poly(\log n)}$ & Quasipolynomial quantum time \\

        \hline
        \multirow{2}{*}{\shortstack[l]{Time $\cT$ \\Clifford+$T$, with at most $\le t$ magic gates}} & $\cT=\poly(n)$, $t=O(\log n)$ & Polynomial time \\
        \cdashline{2-3}
        & $\cT=\poly(n)$, $t=\poly(\log n)$ & Quasipolynomial time \\
        \hline
    \end{tabular}
    \caption{Explicit security for $\cT$-uniform strategies.}
    \label{tab:explicit-security}
\end{table}

For each fixed polynomial bound $\cT=\poly(n)$, the general-channel construction yields a corresponding function computable in polynomial space $\poly(\cT)$. A single family covering all polynomial $\cT$ can instead be obtained by allowing a super-polynomial honest space bound, for example $n^{\log\log\log\log n}$.
Indeed, allowing the adversaries more running time $\cT$ increases the cost of honest evaluation of $f$.

The restricted models permit faster honest evaluation. Indeed, polynomial Pauli sparsity or logarithmically many $T$ gates gives polynomial honest time, while the larger sparsity and $T$-gate budgets in the table give quasipolynomial honest time.

\paragraph{Ebits bounds.}
The number of shared entangled-bits (ebits) is a central benchmark for security against nonlocal quantum computation. In our general model, the parameter $\cT$ bounds each adversary's local qubits, permitting at most $\cT$ shared EPR pairs.
In the $(\cT,s)$-uniform Pauli sparse model, an $s$-term Pauli expansion applied to the all-zero state produces a superposition of at most $s$ product basis vectors, so its Schmidt rank is at most $s$. In particular, sharing $e$ EPR pairs requires $s\geq2^e$, and therefore permits only $O(\log s)$ ebits.
For $(\cT,t)$-uniform Clifford+$T$ strategies, $e$ ebits can be prepared using $2e$ Clifford gates and no magic gates at all. Thus the magic count $t$ alone imposes no additional restriction on their number.

\subsection{Discussion}\label{sec:discussion}

\paragraph{On Our Techniques and Computational Model.}
The closest prior results on lower bounds for explicit $f$-routing are due to \cite{BluhmChristandlSpeelman2022,AsadiCleveCulfMay2025}. The former \cite{BluhmChristandlSpeelman2022} studies the ``bounded-storage model'', where the adversaries are restricted to using $q=q(n)$ many qubits in the purified model.\footnote{Here the strategies are purified into unitaries, so the qubit bound includes registers storing intermediate measurement outcomes coherently. Classical processing of these outcomes is therefore part of the purified strategy.} In this model, they showed security for the inner product function $f= \mathrm{IP}$ as long as $q=O(\log n)$ qubits per adversary. \cite{AsadiCleveCulfMay2025} studied strategies using at most $q$ qubits, and established security for $f = \mathrm{IP}$ as long as the strategies' circuit size is bounded by $\le O(\tfrac{n}{\log q})$ (i.e., the number of quantum gates and measurements).

Importantly, from the bird's eye view, these approaches reduce a successful routing strategy to a communication protocol computing $f$; therefore, the resulting resource lower bounds are inherently limited by the linear upper bound for the communication complexity of any $f$.\footnote{Sending inputs gives upper bounds of $n$ bits for one-way communication and $2n$ bits for simultaneous message passing.} Our new techniques avoid this limitation by using the fidelity gap to reduce a successful strategy into an algorithm computing $f$.

\paragraph{Space-Bounded Adversaries.}
Our reduction preserves space up to constant factors and therefore gives tight security against space-bounded strategies.
Here we assume that the adversaries are space-$\cS$-uniform, namely that all generating and generated machines in \cref{def:uniform-strategies-informal} run in space $\cS$.
By the space-bounded Quantum State Distinguishability algorithm of \citet{LeGallLiuWang2023}, a successful attack implies $f\in\BQSPACE(\cS)$. Under the standard space-constructibility assumptions, the hierarchy for bounded-error quantum space\footnote{Recall that $\BQSPACE(\cS)\subseteq\DSPACE(\cS^2)$ \cite{Watrous2003Space}. If $\BQSPACE(\cS)=\BQSPACE(\cS^{1+\eps})$, repeated padding would give $\BQSPACE(\cS^3)=\BQSPACE(\cS)$, contradicting the deterministic space hierarchy.} \cite{KinneVanMelkebeek2010SpaceHierarchy} yields an explicit total function $f\in\BQSPACE(\cS^{1+\eps})\setminus\BQSPACE(\cS)$ for any fixed $\eps>0$, giving unconditional security with honest quantum space $\cS^{1+\eps}$ against all space-$\cS$-uniform strategies.

\paragraph{Black-hole decoding.}
Our fidelity gap theorem is related to the black-hole decoding problem studied by \citet{HarlowHayden2013}. In this problem, one is given a circuit preparing a state $\rho_{\tt{\bar QHR}}$, where $\tt{H}$ is the remaining black hole, $\tt{R}$ is its emitted radiation, and $\tt{\bar Q}$ is a reference qubit. The task is to implement a decoding channel $\cD_{\tt{R}\to\tt{Q}}$ acting on the radiation and producing a qubit $\tt{Q}$ nearly maximally entangled with $\tt{\bar Q}$, promised that such a channel exists. \citeauthor{HarlowHayden2013} show that an efficient procedure implementing such a decoder for every circuit would imply $\QSZK=\BQP$.

In our case, we study a problem that can be seen as a promise \emph{decision} variant of the approximate black-hole decoding (search) problem. Indeed, given a circuit preparing the midpoint state $\rho^{xy}_{\tt{\bar QAB}}$, regard $\tt{B}$ as the accessible radiation and $\tt{A}$ as the inaccessible system. Our routing promise separates two cases. If $f(x,y)=1$, Bob can distill an approximate EPR pair with $\tt{\bar Q}$. Otherwise, if  $f(x,y)=0$, Alice can similarly distill an approximate EPR pair with $\tt{\bar Q}$, precluding Bob from doing so by monogamy of entanglement.

Put together, the black-hole decoding problem is a search problem, whereas we study a decision variant of it.

\paragraph{Time Hirearchy for $\prBQTIME$.}
For non-explicit uniform Pauli-sparse strategies, our reduction places $f\in\BQTIME((\cT s)^c)$, for which a uniform time hierarchy with polynomial overhead remains unknown for total Boolean functions. It is plausible that the advice-to-promise technique of \citet{He2025PromiseHierarchy}, who proved time-hierarchy theorem for $\prBPTIME$ generalizes to quantum computation, yielding a polynomial time hierarchy for $\prBQTIME$.\footnote{A possible adaptation of \cite{He2025PromiseHierarchy} is as follows. Fix constants $b>a\geq1$. By the quantum time hierarchy theorem with one bit of advice \cite{vanMelkebeekPervyshev2006}, choose a language $L\in\BQTIME(n^b)/1\setminus\BQTIME(n^a)/1$, and let $\alpha_n$ denote the advice bit used by its $O(n^b)$-time algorithm. Consider the promise problem on inputs $(z,\alpha_{|z|})$ with answer $\mathbf{1}_{z\in L}$. The proposed argument is that reading the advice from the input would give a uniform $O(n^b)$-time quantum algorithm, whereas an $O(n^a)$-time algorithm would decide $L$ with one advice bit, contradicting its choice.}
Such extension would imply that for any fixed polynomial resource bounds, our pointwise reduction would yield a partially defined routing function $f\colon\powern\times\powern\to\{0,1,\bot\}$ with polynomial-time bounded-error quantum evaluation on inputs where $f(x,y)\neq\bot$, such that no admissible strategy is $(1-\eps_0)$-correct on all such inputs.
Here $\bot$ denotes inputs outside the promise.

\paragraph{Comparison with CDS.}
Our results permit security against everywhere-correct strategies using even exponentially many qudits $q=2^{O(n)}$, provided their circuit size, description length, and uniform generation time are suitably bounded. We emphasize that this is compatible with the reduction from CDS to $f$-routing of \citet[Corollary~68]{AllerstorferEtAl2024}, which yields a universal attack using $2^{O(\sqrt{n\log n})}$ resource qubits via the CDS protocols of \cite{LiuVaikuntanathanWee2017}.

The reason is that their bound controls quantum resources while allowing unrestricted local computation. The reduction to INDEX uses the truth table $D_y=(f(z,y))_{z\in\powern}$ (see \cite[Section~6.3]{AllerstorferEtAl2024}). In particular, evaluating the CDS message formulas \cite[Figure~5]{LiuVaikuntanathanWee2017} term by term uses $2^n$ evaluations of $f$ during local input processing, before communication. This direct implementation already costs more than a single evaluation of $f(x,y)$, and our model charges for this computation.

\paragraph{Relation to CDQS.}
\citet[Theorem~23]{AllerstorferEtAl2024} established an equivalence between qudit $f$-routing and Conditional Disclosure of Quantum Secrets (CDQS). This equivalence suggests a possible alternative reduction from everywhere-correct $f$-routing strategies to the containment $f\in\QSZK(\cT)$.
On the other hand, our approach exposes a direct relationship between the mid-protocol state and the computational complexity of $f$, by identifying a quantity that determines $f(x,y)$ and prove computational bounds on computing it (in terms of the adversaries' success parameters). We find these bounds as of independent interest, as they quantify the information amount Alice/Bob possess about the routed qudit. Later on, we turn this criterion into efficient algorithms for evaluating $f$ from explicit Pauli-sparse and Clifford+$T$ strategies, translating the adversaries' computational models into security reductions

Another recent paper \cite{AsadiEtAl2025CDQS} also related CDQS to Honest-Verifier $\QSZK$ in \emph{communication complexity}. Their cost measures communication in the proof and its simulation, without bounding the local computation time. These communication bounds are therefore not directly comparable to our time-bounded $\QSZK(\cT)$ statements, which explicitly account for uniform circuit generation and resource-state preparation.

\section{Proof Overview}
For simplicity, throughout this overview we suppose that the strategies are $\cT$-uniform, using at most $q\leq\cT$ qubits (i.e., dimension $d=2$).
We begin with the fidelity bounds that underlie our reductions, describing first the protocol's state development (see \cref{sec:setup} for more details).
Initially, Alice and Bob share the following state, with Alice holding system $\mathtt L$ and Bob holding $\mathtt R$:
\[
\cI_{\tt{LR} \to \tt{LR}}\ps{\ketbraa{0^{2q}}_{\mathtt{LR}}}.
\]
Here $\cI$ is a quantum channel describing the preparation of the initial state from the all-zero ancilla.

Then, the Verifiers $V_0,V_1$ prepare the EPR state $\Phi_{\tt{\Qb Q}}$, and sample $x,y \inr \powern$, respectively. $V_0$ sends Alice the challenge register $\Phi_{\mathtt Q}$ and $x$, whereas $V_1$ sends $y$ to Bob.
Alice and Bob now apply their local strategies, each relying only on their own input.
Specifically, Alice applies the channel $\cU^x_{\mathtt{QL}\to \tt{A_1 B_2}}$ and Bob applies $\cV^y_{\tt{R} \to \tt{B_1 A_2}}$, therefore preparing the mid-protocol state
\[
\rhopre_{\ol{\mathtt Q}\mathtt{A_1B_2 B_1A_2}}
=\ps{\operatorname{id}_{\ol{\mathtt Q}}\otimes\cU^x\otimes\cV^y}
\ps{\Phi_{\ol{\mathtt Q}\mathtt Q}\otimes\cI^n\ps{\ketbraa{0^{2q}}_{\mathtt{LR}}}}.
\]
Here $\tt{B_2}$ and $\tt{A_2}$ are the subsystems Alice and Bob intend to send to each other, respectively. Each chooses the partition into retained and outgoing subsystems using only their local input.
After the exchange, we denote Alice's and Bob's systems by $\tt{A}=\tt{A_1 A_2}$ and $\tt{B}=\tt{B_1 B_2}$, respectively, and write the mid-protocol state as $\rhopre_{\tt{\Qb AB}}$.

After exchanging their systems and classical inputs, Alice and Bob both know $(x,y)$. We allow them to use $f(x,y)$ in their final local computation, just as the honest prover uses it to choose where to send the challenge. They then apply their final decoding channels $\cU^{x,y,f(x,y)}_{\mathtt A\to\mathtt{Q}_0}$ and $\cV^{x,y,f(x,y)}_{\mathtt B\to\mathtt{Q}_1}$, respectively, producing the final-protocol state
\[
\rho^{xyf(x,y)}_{\ol{\mathtt Q}\mathtt{Q}_0\mathtt{Q}_1} = 
\ps{ \mathrm{id}_{\tt{\Qb}} \otimes \cU^{x,y,f(x,y)}_{\mathtt A\to\mathtt{Q}_0} \otimes \cV^{x,y,f(x,y)}_{\mathtt B\to\mathtt{Q}_1} }
(\rhopre_{\ol{\mathtt Q}\mathtt{AB}}).
\]
Correctness requires that depending on the value $b=f(x,y)$, either Alice or Bob should be able to output a qudit that is highly entangle with the reference system, that is:
\begin{align*}
    f(x,y) = 0 &\qlq F(\rhopost_{\tt{\Qb Q_0}} , \Phi_{\tt{\Qb Q_0}}) \ge 1-\eps \\
    f(x,y) = 1 &\qlq F(\rhopost_{\tt{\Qb Q_1}} , \Phi_{\tt{\Qb Q_1}}) \ge 1-\eps.
\end{align*}
We emphasize that each adversary's pre-communication channel ($\cU^x, \cV^y$), including which registers to exchange with each other (i.e., the partition $\tt{A_1,B_2}$ and $\tt{B_1,A_2}$), depends only on its local input $x,y$. Thus, given $(x,y)$, we can construct a circuit preparing $\rhopre_{\tt{\Qb AB}}$ without computing $f(x,y)$ itself.

Now, suppose the strategy is $(1-\eps)$-correct on input $(x,y)$, and fix the mid-protocol state $\rhopre = \rhopre_{\tt{\Qb AB}}$. We compare each adversary's joint state with the reference to the product of its marginals by defining:
\[
F_0\defeqq F\ps{\rho_{\ol{\mathtt Q}\mathtt A},\rho_{\ol{\mathtt Q}}\otimes\rho_{\mathtt A}},
\qquad
F_1\defeqq F\ps{\rho_{\ol{\mathtt Q}\mathtt B},\rho_{\ol{\mathtt Q}}\otimes\rho_{\mathtt B}}.
\]
Intuitively, $F_0$ and $F_1$ measure how close each adversary's joint state with the reference is to the product of its marginals. We show that, depending on the value of $f(x,y)$, one quantity is close to 1 while the other is bounded away from 1. The intended recipient can recover entanglement with the reference, while the other adversary's joint state with the reference is approximately a product state.
For sufficiently small $\eps$, these bounds leave a constant gap, so approximating $F_0$ to sufficiently small constant additive error determines $f(x,y)$. This reduces computing $f$ to estimating the fidelity of states prepared from the strategy.

\subsection{Fidelity Bounds}\label{sec:overview-fidelity-bounds}
\paragraph{Lower bound.}
Suppose $f(x,y)=1$ and the adversaries route the qubit correctly with probability at least $1-\eps$. Bob can then recover a qubit that is nearly maximally entangled with the reference. We show that this forces Alice's joint state with the reference to be approximately a product state, namely 
\[
\rhopre_{\ol{\mathtt{Q}} \mathtt{A}} \approx \rhopre_{\ol{\mathtt{Q}}} \otimes \rhopre_{\mathtt{A}}.\footnote{Throughout this overview, we say that two states are approximately the same $\Psi_1 \approx_\delta \Psi_2$ if their fidelity is at least $F(\Psi_1 , \Psi_2) \ge 1- \delta$.}
\]
The key point is that an EPR pair is pure, so Bob's recovery of an approximate EPR pair leaves little room for correlations with Alice. We make this precise by applying Bob's decoder and using Uhlmann's theorem. Tracing out Bob's system then gives a product approximation to Alice's unchanged joint state with the reference. 

We now formalize this argument. We first apply  Bob's post-communication decoding channel $\cV_{\mathtt B\to\mathtt{Q}_1}$ to the mid-protocol state (leaving Alice's system and reference untouched):
\begin{align*}
\sigma_{\ol{\mathtt Q}\mathtt{Q}_1\mathtt A}
&\defeqq
\ps{\operatorname{id}_{\ol{\mathtt Q}\mathtt A}\otimes\cV_{\mathtt B\to\mathtt{Q}_1}}
\ps{\rhopre_{\ol{\mathtt Q}\mathtt{AB}}}.
\end{align*}
Correctness guarantees that $\sigma_{\ol{\mathtt Q}\mathtt{Q}_1}$ is close to an EPR pair, so Uhlmann's theorem implies that $\sigma$ is close to this EPR pair tensored with some state $\tau_{\mathtt A}$ on Alice's register:
\[
\sigma_{\ol{\mathtt{Q}} \mathtt{Q}_1\mathtt{A}}
\approx_\eps
\Phi_{\ol{\mathtt{Q}} \mathtt{Q}_1}\otimes\tau_{\mathtt{A}}.
\]
Since Bob's decoding channel is trace preserving and acts only on his system, it leaves Alice's joint state with the reference unchanged. Thus, $\sigma$ and $\rho$ have the same reduced state on $\ol{\mathtt Q}\mathtt A$, and discarding Bob's decoded register in the comparison above shows that this state is close to a product state:
\[
\rho_{\ol{\mathtt{Q}} \mathtt{A}}
    =
    \tr_{\mathtt{Q}_1}\!\left(\sigma_{\ol{\mathtt{Q}} \mathtt{Q}_1\mathtt{A}}\right)
    \approx_\eps
    \tr_{\mathtt{Q}_1}\!\left(
        \Phi_{\ol{\mathtt{Q}} \mathtt{Q}_1}\otimes\tau_{\mathtt{A}}
    \right)
    =
    \rho_{\ol{\mathtt{Q}}}\otimes\tau_{\mathtt{A}}.
\]
Taking another marginal and discarding the reference system $\ol{\mathtt{Q}}$ shows that $\tau_{\mathtt{A}}$ must itself be close to Alice's actual marginal $\rho_{\mathtt{A}}$:
\[
\rho_{\mathtt{A}}
    =
    \tr_{\ol{\mathtt{Q}}}\!\left(\rho_{\ol{\mathtt{Q}} \mathtt{A}}\right)
    \approx_\eps
    \tr_{\ol{\mathtt{Q}}}\!\left(
        \rho_{\ol{\mathtt{Q}}}\otimes\tau_{\mathtt{A}}
    \right)
    =
    \tau_{\mathtt{A}}.
\]
Thus we may replace $\tau_{\mathtt{A}}$ by $\rho_{\mathtt{A}}$, and the fidelity-angle triangle inequality shows that $\rho_{\ol{\mathtt{Q}} \mathtt{A}}$ is close to the product of its marginals:
\[
\rho_{\ol{\mathtt{Q}}}\otimes\tau_{\mathtt{A}}
    \approx_\eps
    \rho_{\ol{\mathtt{Q}}}\otimes\rho_{\mathtt{A}}.
\]
In short, Bob's successful recovery of the EPR pair forces Alice to be approximately decoupled from $\ol{\mathtt{Q}}$, and combining everything together yields
\begin{align*}
    \rho_{\ol{\mathtt{Q}} \mathtt{A}}
    \approx_{4\eps}
    \rho_{\ol{\mathtt{Q}}}\otimes\rho_{\mathtt{A}}.
\end{align*}

\paragraph{Upper Bound.}
We now complement the result with an upper bound on the fidelity of Bob's joint system. Suppose that $f(x,y)=1$, and suppose that the adversaries successfully route the qudit to Bob. We would like to show that Bob's mid-protocol state $\rho_{\ol{\mathtt{Q}} \mathtt{B}} = \tr_{\mathtt{A}}\pss{\rhopre_{\ol{\mathtt{Q}} \mathtt{AB}}}$ is highly non-product across the cut $\ol{\mathtt{Q}}:\mathtt{B}$, in the sense that
\begin{align*}\label{eq:refefad}
    F\ps{
        \rho_{\ol{\mathtt{Q}} \mathtt{B}} \; , \;
        \rho_{\ol{\mathtt{Q}}}\otimes\rho_{\mathtt{B}}
    }
    \leq
    \frac{1}{d}+\sqrt{2\eps}.
\end{align*}
To see this, apply Bob's decoder to both the ``true state'' $\rho_{\ol{\mathtt{Q}} \mathtt{B}}$ and the corresponding ``product state'' $\rho_{\ol{\mathtt{Q}}}\otimes\rho_{\mathtt{B}}$. On the true state, correctness says that the decoder recovers a state $\sigma_{\ol{\mathtt{Q}} \mathtt{Q}_1}$ that has high fidelity with the EPR state:
\[
    \rho_{\ol{\mathtt{Q}} \mathtt{B}}
    \quad
    \xrightarrow{\quad \cV^{xyb}_{\mathtt{B}\to \mathtt{Q}_1}\quad}
    \quad 
    \sigma_{\ol{\mathtt{Q}} \mathtt{Q}_1}
    \approx_{\eps}
    \Phi_{\ol{\mathtt{Q}} \mathtt{Q}_1}.
\]
Consequently, an EPR measurement on $\sigma_{\ol{\mathtt{Q}} \mathtt{Q}_1}$ accepts with high probability
\begin{align}\label{eq:nice f1}
    \tr(\Phi\sigma)
    =
    F(\sigma,\Phi)^2
    \geq
    (1-\eps)^2
    \geq
    1-2\eps.
\end{align}
Now apply the same decoder to the product state. Because Bob's decoder acts only on $\mathtt{B}$, it cannot create correlations with the untouched reference register $\ol{\mathtt{Q}}$. Hence its output remains a product state, denoted $\omega$:
\[
    \rho_{\ol{\mathtt{Q}}}\otimes\rho_{\mathtt{B}}
    \quad
    \xrightarrow{\quad \cV^{xyb}_{\mathtt{B}\to \mathtt{Q}_1}\quad}
    \quad
    \omega_{\ol{\mathtt{Q}} \mathtt{Q}_1}
    =
    \rho_{\ol{\mathtt{Q}}}\otimes\omega_{\mathtt{Q}_1}
    =
    \frac{I_{\ol{\mathtt{Q}}}}{d}\otimes\omega_{\mathtt{Q}_1}.
\]
Such a product state passes the EPR measurement with probability exactly
\begin{align}\label{eq:nice f2}
    \tr(\Phi\omega)
    =
    \frac{1}{d^2}.
\end{align}
Consequently, the two decoded states behave very differently under the EPR measurement (\cref{eq:nice f1,eq:nice f2}), and hence their fidelity is bounded $F\pss{\sigma, \omega} \le \frac{1}{d} + \sqrt{2\eps}.$
To transfer this upper bound to the mid-protocol states, we use the data-processing inequality for fidelity. Applying Bob's decoder to both states cannot decrease their fidelity, so
\begin{align*}
F\ps{\rho_{\ol{\mathtt{Q}} \mathtt{B}} \;, \; \rho_{\ol{\mathtt{Q}}} \otimes \rho_{\mathtt{B}}}
\le  F\pss{\sigma, \omega}
\le \frac{1}{d} + \sqrt{2\eps}.
\end{align*}
That implies the true state cannot be close to the product of its marginals: successful recovery by Bob requires substantial correlations between $\mathtt{B}$ and the reference register $\ol{\mathtt{Q}}$.

\subsection{Faster Algorithms for Computationally Bounded Adversaries}
In \cref{sec:overview-fidelity-bounds}, we reduced computing $f(x,y)$ from a successful strategy to estimating the mid-protocol fidelity $F_0$. The two states in this fidelity can be prepared by running the adversaries' resource preparation and pre-communication circuits, with only polynomial overhead. For arbitrary channels, this gives a $\QSZK$ protocol whose verifier time is polynomial in the strategy's resources, including the time needed to generate its circuits (\cref{thm:routing-to-qszk}).

Our goal now is to obtain tighter security bounds by exploiting restrictions on these circuits to compute the fidelity directly. We consider Pauli-sparse and $(\cT,t)$-uniform Clifford+$T$ strategies, which restrict the implementations of the resource preparation $\cI$ and the local channels $\cU^x,\cV^y$. In each case, their structure lets us reduce the fidelity computation to a small matrix. The resulting algorithm, combined with the fidelity gap, gives a time bound for computing $f$ from a successful strategy, and hence security for functions that require more time.

\paragraph{Clifford+$T$ Adversaries.}
We consider $(\cT,t)$-uniform Clifford+$T$ strategies. To analyze these strategies, we study states prepared by Clifford+$T$ circuits on $q \le \cT$ qubits, with at most $t$ magic gates (recall we suppose these circuits have size at most $\cT$ for simplicity).
% 
% To begin with, we emphasize that the rank bound from the Pauli-sparse case no longer applies: even with no magic gates, a Clifford circuit can prepare many EPR pairs across the cut, giving exponentially many Schmidt coefficients. Nevertheless, the bound on the number of magic gates lets us separate out fixed qubits, leaving a residual state on only $O(t)$ qubits. 
% 
At a high level, we separate out fixed factors and reduce the fidelity computation to a calculation on a logical state of only $O(t)$ qubits. This gives a deterministic algorithm running in time $\poly(2^t,\cT)$ for constant additive error.

To obtain this separation and compute the residual state, we introduce an $(m,\nu)$-stabilizer decomposition for a Clifford+$T$-prepared state $\ket\psi$, consisting of the following two kinds of known decompositions.
\begin{enumerate}
    \item
    First, we represent $\ket\psi$ using an $m$-stabilizer decomposition, following the approach of \citet{BravyiEtAl2019}. Such a decomposition expresses the state as a linear combination of $m$ stabilizer states,
\[
    \ket{\psi}=\sum_{j=1}^{m}\alpha_j\ket{s_j}.
\]
With at most $t$ magic gates, we can construct such a decomposition with $m\leq2^t$. Indeed, each magic gate is a linear combination of two Clifford gates, and each Clifford gate maps a stabilizer state to a stabilizer state. Thus each magic gate at most doubles the number of stabilizer terms, while Clifford gates leave it unchanged (see \cref{lem:cliffordT-nullity}).

\item
Second, we use the stabilizer notion nullity introduced by \citet{BeverlandEtAl2020}. The Pauli operators that fix $\ket\psi$ form its stabilizer group. For a $q$-qubit stabilizer state, this group has $q$ independent generators. However, an arbitrary pure state may have fewer independent generators. This deficiency is called stabilizer nullity. We record $q-\nu$ independent Pauli operators satisfying
\[
    P_j\ket\psi=\ket\psi,
    \qquad j\in[q-\nu].
\]
With at most $t$ magic gates, we can efficiently find such operators with $\nu\leq t$. Clifford gates preserve the number of independent stabilizers, while each magic gate loses at most one (\cref{lem:cliffordT-nullity}).
\end{enumerate}

We use this decomposition as follows: we efficiently construct a state $\ket \chi$, that on the one hand admits $(m,\nu)$-decomposition, and on the other, its Schmidt coefficients capture the target fidelity. Let $\ket{\rho}_{\mathtt{AE}_\rho}$ and $\ket{\sigma}_{\mathtt{AE}_\sigma}$ be pure states prepared by Clifford+$T$ circuits, each with at most $t$ magic gates. We want to compute the fidelity of their reduced states on $\mathtt A$,
\[
    \rho=\tr_{\mathtt E_\rho}\ketbraa{\rho},
    \qquad
    \sigma=\tr_{\mathtt E_\sigma}\ketbraa{\sigma}.
\]
Using the identity $F(\rho,\sigma)=\norm {\tr_{\mathtt A}\ps{\ketbra{\sigma}{\rho}}}_1,$ we reduce computing the fidelity to computing the singular values of
\[
M = \tr_{\mathtt A}\ps{\ketbra{\sigma}{\rho}}.
\]
The matrix $M$ can still be exponentially large in $q$. To this end, we regard its entries as the amplitudes of a bipartite vector, whose Schmidt coefficients are exactly the singular values of $M$, by defining
\[
    \ket{\chi}_{\mathtt E_\sigma\mathtt E_\rho}
    \defeqq\sum_{a,b}\langle b|M|a\rangle
    \ket b_{\mathtt E_\sigma}\ket a_{\mathtt E_\rho},
\]
where $\{\ket a\}$ and $\{\ket b\}$ are computational bases of the environment systems $\mathtt E_\rho$ and $\mathtt E_\sigma$, respectively.

The partial trace defining $M$ corresponds to projecting the two $\mathtt A$ registers of $\ket\sigma\otimes\overline{\ket\rho}$ onto EPR pairs and keeping the environment registers. 
In \cref{lem:bell-contraction-data}, we use the decompositions of the two purifications, and show how to construct an explicit $(m_\chi,\nu_\chi)$-decomposition with $m_\chi\leq2^{2t}$ and $\nu_\chi\leq2t$.

Our next goal is to reduce the fidelity computation to a small logical state. We use the two parts of the decomposition in turn. The nullity bound $\nu_\chi\leq2t$ will let us isolate this logical state, and the expansion into $m_\chi$ stabilizer states will then let us compute its amplitudes.

To see why small nullity helps, write $q_\chi$ for the number of qubits of $\ket\chi$. The decomposition gives $q_\chi-\nu_\chi$ independent Pauli constraints on this state. A suitable Clifford change of basis turns each constraint into the requirement that a distinct qubit be $\ket0$ (\cref{thm:clifford-reduction}), leaving only $\nu_\chi$ logical qubits to describe. However, this change of basis may act jointly on the two environment registers and alter the Schmidt coefficients of $\ket\chi$ across $\mathtt E_\sigma:\mathtt E_\rho$, whose sum is the target fidelity. To recover the fidelity, we need a reduction that preserves these coefficients.

We therefore use Clifford circuits acting separately on the two environment registers. Such local changes of basis preserve the Schmidt coefficients. In \cref{lem:low-nullity-form}, we show that they separate out fixed states, leaving a logical state on at most $2\nu_\chi\leq4t$ qubits. This gives efficiently computable circuits $D_{\mathtt E_\sigma},D_{\mathtt E_\rho}$ acting on the respective sides, with
\begin{align}\label{eq:overview-chi-normal-form}
    (D_{\mathtt E_\sigma}\otimes D_{\mathtt E_\rho})\ket{\chi}
    &=\ket{0}^{\otimes r}\otimes\ket{\Phi^+}^{\otimes p}\otimes\ket{\gamma},
\end{align}
where $\ket{\gamma}$ occupies only $O(t)$ qubits, and $\ket{\Phi^+}$ is the EPR state.

The zero states and EPR pairs contribute known factors, so it remains to compute the sum of the Schmidt coefficients of $\ket\gamma$.

The small size of $\ket\gamma$ is useful only if we can also compute its amplitudes. For this final step, we return to the explicit expansion of $\ket\chi$ into at most $m_\chi \le 2^{2t}$ stabilizer states. We use stabilizer calculations on each term to recover the amplitudes of $\ket\gamma$ in time $\poly(2^t,\cT)$.

Finally, since $\ket\gamma$ occupies only $O(t)$ qubits, its coefficient matrix across the $\mathtt E_\sigma:\mathtt E_\rho$ cut has size $2^{O(t)}$. Approximating its singular values gives the Schmidt coefficients to the required precision in time $\poly(2^t,\cT)$.
Combining their sum with the contribution of the fixed factors in \cref{eq:overview-chi-normal-form} gives the desired fidelity estimate within the same time bound.

\paragraph{Pauli Sparse Unitaries.}
We consider unitaries with sparse Pauli decomposition.
Let $U,V$ be two $(q,s)$-Pauli-sparse unitaries on $q$ qubits, and suppose they both act on two subsystems $\tt{A, \tt{\ol A}}$. Let $\rho,\sigma$ be the reduced states on $\mathtt A$ of $U\ket{0^q},V\ket{0^q}$, respectively. Our goal is to reduce the fidelity computation $F(\rho, \sigma)$ to a matrix whose size depends on $s$. We first assume that the Pauli decompositions of $U,V$ are given explicitly, and later obtain them by learning. 

To see how sparsity helps, decompose $U$ into a linear sum of Pauli matrices:
\[
    U=\sum_{j=1}^{s}\alpha_j P_j.
\]
Observe that each Pauli $P_j$ maps $\ket{0^q}$ to a computational-basis vector up to a phase. Absorbing these phases into coefficients $c_j$ gives
\[
    U\ket{0^q}=\sum_{j=1}^{s}c_j
    \ket{a_j}_{\mathtt A}\ket{b_j}_{\bar{\mathtt A}}.
\]
This is a sum of at most $s$ product vectors, so its Schmidt rank across $\mathtt A:\bar{\mathtt A}$ is at most $s$. To turn this rank bound into an explicit factorization of $\rho$, we group terms with the same $b_j$, giving
\[
    U\ket{0^q}=\sum_b\ket{u_b}_{\mathtt A}\ket b_{\bar{\mathtt A}},
    \qquad \ket{u_b}=\sum_{j:b_j=b}c_j\ket{a_j}.
\]
Tracing out $\bar{\mathtt A}$ now gives $\rho=XX^\dagger$, where the columns of $X$ are the at most $s$ vectors $\ket{u_b}$. Applying the same construction to $V$ gives $\sigma=YY^\dagger$ with at most $s$ columns as well. These factorizations reduce the fidelity computation to
\[
    F(\rho,\sigma)=\norm{X^\dagger Y}_1.
\]
The sparse descriptions let us construct the at-most-$s\times s$ matrix $X^\dagger Y$ in time $\poly(q,s)$. Approximating the sum of its singular values therefore gives the desired fidelity estimate in roughly the same time.

When the Pauli decompositions are not given explicitly, we use a quantum learning algorithm to identify the significant Pauli coefficients and estimate their amplitudes. After discarding small coefficients and normalizing the resulting vectors, we obtain explicit sparse descriptions to which the same fidelity computation applies. Choosing the learning error sufficiently small gives a constant-additive-error fidelity estimate with bounded error in time $\poly(\cT,s)$.

\section{Preliminaries}

\subsection{Quantum Information}

\begin{lemma}[{Fuchs--van de Graaf; \cite[Eq.~(9.110)]{NC00}}]\label{lem:fvg}
    For states $\rho,\sigma$,
    \begin{align*}
        1 - F\ps{\rho,\sigma} \le \frac{1}{2}\norm{\rho-\sigma}_1 \le \sqrt{1 - F\ps{\rho,\sigma}^2} .
    \end{align*}
\end{lemma}

% Triangle inequality for the fidelity:
\begin{theorem}[{\cite[Sec.~9.2.2]{NC00}}]\label{thm:fidelity triangle}
    For all states $\rho, \omega, \sigma$,
    \begin{align*}
        F\ps{\rho,\sigma} \ge F\ps{\rho,\omega} F\ps{\omega,\sigma} - \sqrt{1 - F\ps{\rho,\omega}^2} \sqrt{1 - F\ps{\omega,\sigma}^2}.
    \end{align*}
\end{theorem}
\begin{proof}
    Define the fidelity angles
    \[
    a:=\arccos F(\rho,\omega),\qquad
    b:=\arccos F(\omega,\sigma),\qquad
    c:=\arccos F(\rho,\sigma).
    \]
    Since fidelity takes values in $[0,1]$, we have
    $a,b,c\in[0,\pi/2]$. The triangle inequality for the fidelity angle
    gives
    \[
    c\leq a+b.
    \]
    Since $a+b\in[0,\pi]$ and cosine is decreasing on this interval,
    \[
    F(\rho,\sigma)
    =
    \cos c
    \geq
    \cos(a+b).
    \]
    Using the cosine addition formula concludes the claim:
    \begin{align*}
        \cos(a+b)
        &=
        \cos a\cos b-\sin a\sin b\\
        &=
        F(\rho,\omega)F(\omega,\sigma)
        -
        \sqrt{1-F(\rho,\omega)^2}
        \sqrt{1-F(\omega,\sigma)^2}.
    \end{align*}
\end{proof}

\begin{fact}[{\cite[Theorem~3.27]{Watrous18TQI}}]\label{thm:fidelity data processing}
    The Fidelity of two states is nondecreasing under any quantum channel $\cE$:
    \[
    F\ps{\cE(\rho), \cE(\sigma)} \geq F\ps{\rho, \sigma}.
    \]
\end{fact}

\begin{fact}[{\cite[Eq. (1.132)]{Watrous18TQI}}]\label{thm:epr measurement}
    For any two operators $A,B$:
    \begin{align*}
        \bra\Phi \ps{A \otimes B} \ket\Phi
        = \frac1d \sum_{i,j} \bra{i} A \ket{j} \bra{i} B \ket{j}
        = \frac1d \sum_{i,j} A_{ij} \ps{B^{T}}_{ji}
        = \frac1d \tr\ps{A B^{T}}.
    \end{align*}
\end{fact}

\subsection{Norms}
We write $\norm{x}_2$ for the Euclidean norm of a vector. For matrices, we use the trace norm $\norm{A}_1\defeq\tr\sqrt{A^\dagger A}$ and the Frobenius norm $\norm{A}_F\defeq(\sum_{i,j}\abs{A_{ij}}^2)^{1/2}$. 

\subsection{Numerical Analysis}

We make use the following standard fact about computing singular values.
\begin{theorem}[{\cite{WikipediaSVD}}]\label{thm:svd}
    There is a deterministic algorithm that, given $a, b \in \N$, an accuracy $\delta \in (0,1)$, and the entries of a matrix $M \in \C^{a \times b}$ as rationals of bit length at most $\ell$, outputs the singular values of $M$, and in particular $\norm{M}_1$, up to additive error $\delta$, using $\poly\ps{a, b, \ell, \log\frac{1}{\delta}}$ bit operations.
\end{theorem}

The SVD polynomial time complexity  is folklore in the numerical linear algebra literature. We complement it with a brief accounting for the bit complexity of truncation and the resulting approximation error.

Let $r=\min(a,b)$, and let $\wt M$ be obtained by truncating each entry of $M$ to absolute error at most $\eta=\delta/(2r\sqrt{ab})$. Writing $\sigma_i$ for the $i$th largest singular value, we have, for every $i\leq r$,
\[
|\sigma_i(M)-\sigma_i(\wt M)|
\leq \sqrt{\sum_{j=1}^{a}\sum_{k=1}^{b}|M_{jk}-\wt M_{jk}|^2}
\leq \sqrt{ab}\,\eta
=\frac{\delta}{2r},
\]
Thus, computing each $\sigma_i(\wt M)$ to additive error $\delta/(2r)$ gives error at most $\delta/r$ per singular value and at most $\delta$ in their sum, $\norm{M}_1$. The truncation requires only $O(\log(ab/\delta))$ fractional bits per entry.

\subsection{Setup}\label{sec:setup}

We first describe the routing strategy and the notation for its state at each stage of the protocol.  This setup is common to all of our results.  The computational assumptions on the adversaries' circuits will be specified separately in the corresponding sections.

\begin{definition}[$\cT$-uniform $f$-routing strategies]\label{def:uniform-routing-strategy}
Let $d,\cT\colon\N\rightarrow\N$ be time-constructible, nondecreasing functions, with $d(n)\geq2$ and $\cT(n)\geq n$.  All registers are $d(n)$-dimensional qudits.  For every $n$, the strategy fixes systems
\[
\ol{\mathtt{Q}},\mathtt{Q},\mathtt{Q}_0,\mathtt{Q}_1,\mathtt{L},\mathtt{R},\mathtt{A}_1,\mathtt{A}_2,\mathtt{B}_1,\mathtt{B}_2,
\]
and we write
\[
\mathtt{A}\defeq\mathtt{A}_1\mathtt{A}_2, \qquad \text{and,} \qquad \mathtt{B}\defeq\mathtt{B}_1\mathtt{B}_2.
\]
A $\cT$-uniform $f$-routing strategy is specified by five fixed, uniform, deterministic Turing machines which generate the resource preparation and Alice's and Bob's pre- and post-communication circuits, denoted by
\[
    (\genresource,\genalice,\genbob,\genalicepost,\genbobpost).
\]
The generators satisfy the following properties:
\begin{enumerate}
    \item \emph{Resource-state preparation.}
    For every $n$, machine $\genresource$ outputs the description of a quantum channel $\cI^{n}$ preparing the adversaries' initial entangled resource state.
    \[
        \genresource(1^n)\mapsto\cI^n_{\mathtt{LR}\to\mathtt{LR}}.
    \]

    \item \emph{Adversarial circuits.}
    For every $n\in \N$ and $x,y\in \powern$, the machines $\genalice$ and $\genbob$ output descriptions of Alice's and Bob's pre-communication channels, respectively:
    \[
        \genalice(1^n,x)\mapsto\cU^x_{\mathtt{QL}\to\mathtt{A}_1\mathtt{B}_2},
        \qquad
        \genbob(1^n,y)\mapsto\cV^y_{\mathtt{R}\to\mathtt{B}_1\mathtt{A}_2}.
    \]

    \item \emph{Post-communication channels.}
    For every $n\in\N$, $x,y\in\powern$, and $b\in\power{}$, the machines $\genalicepost$ and $\genbobpost$ output descriptions of Alice's and Bob's post-communication channels, respectively:
    \[
        \genalicepost(1^n,x,y,b)\mapsto
        \cU^{xyb}_{\mathtt{A}\to\mathtt{Q}_0},
        \qquad
        \genbobpost(1^n,x,y,b)\mapsto
        \cV^{xyb}_{\mathtt{B}\to\mathtt{Q}_1}.
    \]
    \item 
    For every $n\in \N$, each of the three generator machines $\genresource,\genalice,\genbob$ runs in time at most $\cT(n)$ (whereas the post-communication generation machines are unrestricted).
\end{enumerate}
\end{definition}

\paragraph{The Protocol's State.}
Let us now describe the protocol's state development. Fix $n$, and write $q= q(n)\leq \cT(n)$, where $2q$ is the number of qudits used by the initialization channel. See also \Cref{fig:main-nlqc} for illustration.
\begin{enumerate}
    \item The initial state consists of an EPR pair and the adversaries' resource state, prepared on registers $\mathtt{L}$ and $\mathtt{R}$:
    \[
    \rhoinit_{\ol{\mathtt{Q}}\mathtt{QLR}} = \Phi_{\ol{\mathtt{Q}}\mathtt{Q}} \otimes
    \cI^n_{\mathtt{LR}\to\mathtt{LR}} \ps{\ketbraa{0^{2q}}_{\mathtt{LR}}}.
    \]
    
    \item The challenges $x,y\inr\powern$ are sampled.
    
    \item Alice and Bob locally apply their pre-communication channels, producing the mid-protocol state
    \[
    (\rhopre)_{\ol{\mathtt{Q}}\mathtt{AB}} =
    \ps{\operatorname{id}_{\ol{\mathtt{Q}}}\otimes
    \cU_{\mathtt{QL}\to\mathtt{A}_1\mathtt{B}_2}^{x}
    \otimes \cV^{y}_{\mathtt{R}\to\mathtt{B}_1\mathtt{A}_2}}\ps{\rhoinit}.
    \]
    Alice sends $\mathtt{B}_2$ to Bob, and Bob sends $\mathtt{A}_2$ to Alice.  After this exchange, Alice holds $\mathtt{A}=(\mathtt{A}_1,\mathtt{A}_2)$ and Bob holds $\mathtt{B}=(\mathtt{B}_1,\mathtt{B}_2)$.

    \item 
    For $b=f(x,y)$, Alice and Bob apply the corresponding local post-communication channels, producing the final state
    \[
    (\rhopost)_{\ol{\mathtt{Q}}\mathtt{Q}_0\mathtt{Q}_1}=
    \ps{\operatorname{id}_{\ol{\mathtt{Q}}}\otimes
    \cU^{xyb}_{\mathtt{A}_1\mathtt{A}_2\to\mathtt{Q}_0} \otimes
    \cV^{xyb}_{\mathtt{B}_1\mathtt{B}_2\to\mathtt{Q}_1}}\ps{(\rhopre)_{\ol{\mathtt{Q}}\mathtt{AB}}}.
    \]
\end{enumerate}

\begin{figure}[H]
\centering
\begin{tikzpicture}[
  x=0.68cm,y=0.39cm,>=stealth,font=\footnotesize,
  channel/.style={draw,fill=white,minimum width=0.96cm,minimum height=0.48cm,inner sep=1.3pt},
  wire/.style={->,semithick,shorten >=2pt}
]
  \draw[->] (-1.3,-1.35) -- (-1.3,14.55) node[above] {time};
  \draw[->] (-1.3,-1.35) -- (11.1,-1.35) node[right] {space};
  \foreach \x in {0,10} {
    \draw[gray,densely dotted] (\x,-0.1) -- (\x,14.1);
  }
  % Position guides stop at the operation layers to keep the wires distinct.
  \foreach \x in {3,7} {
    \draw[gray,densely dotted] (\x,-0.1) -- (\x,3.3);
    \draw[gray,densely dotted] (\x,11.1) -- (\x,14.1);
  }
  \node[below] at (0,-0.15) {$V_0$};
  \node[below] at (3,-0.15) {Alice};
  \node[below] at (7,-0.15) {Bob};
  \node[below] at (10,-0.15) {$V_1$};

  \node[channel] (alicepre) at (3,4.1) {$\cU^x$};
  \node[channel] (bobpre) at (7,4.1) {$\cV^y$};
  \node[channel] (alicepost) at (3,10.4) {$\cU^{xyb}$};
  \node[channel] (bobpost) at (7,10.4) {$\cV^{xyb}$};

  % Separate ports for the challenge and the entangled resource.
  \coordinate (aQ) at ([xshift=-0.23cm]alicepre.south);
  \coordinate (aL) at ([xshift=0.23cm]alicepre.south);
  \coordinate (bR) at ([xshift=-0.23cm]bobpre.south);
  \coordinate (by) at ([xshift=0.23cm]bobpre.south);
  \coordinate (aKeep) at ([xshift=-0.23cm]alicepre.north);
  \coordinate (aSend) at ([xshift=0.23cm]alicepre.north);
  \coordinate (bSend) at ([xshift=-0.23cm]bobpre.north);
  \coordinate (bKeep) at ([xshift=0.23cm]bobpre.north);
  \coordinate (aKeepIn) at ([xshift=-0.23cm]alicepost.south);
  \coordinate (aReceive) at ([xshift=0.23cm]alicepost.south);
  \coordinate (bReceive) at ([xshift=-0.23cm]bobpost.south);
  \coordinate (bKeepIn) at ([xshift=0.23cm]bobpost.south);

  \fill (0,0.6) circle (1.2pt);
  \fill (10,0.6) circle (1.2pt);
  \draw[wire] (0,0.6) -- (aQ)
    node[pos=0.43,above left,inner sep=1pt] (alicechallenge) {$x,\mathtt Q$};
  \draw[wire] (10,0.6) -- (by)
    node[pos=0.43,above right,inner sep=1pt] (bobchallenge) {$y$};
  % The retained reference propagates upward along the verifier worldline.
  \draw[->,semithick] (0,0.6) -- (0,12.9);
  \node[anchor=east] at (-0.12,7.3) {$\ol{\mathtt Q}$};

  \draw[semithick] (aL)
    .. controls (3.34,0.15) and (6.66,0.15) .. (bR);
  \node (resourceprep) at (5,0.55) {$\mathcal{I}$};
  % Initial state before resource preparation.
  \draw[densely dotted,thin] (-0.3,-0.05) -- (10.3,-0.05);
  \node[anchor=west] at (10.45,-0.05) {$\rho_{\ol{\mathtt Q}\mathtt Q}\otimes |0\rangle^q$};
  \node[anchor=east] (resourceleft) at (3.4,2.55) {$\mathtt L$};
  \node[anchor=west] (resourceright) at (6.6,2.55) {$\mathtt R$};

  \draw[wire] (aKeep) -- (aKeepIn)
    node[pos=0.43,left,inner sep=2pt] (alicekeeplabel) {$\mathtt A_1$};
  \draw[wire] (bKeep) -- (bKeepIn)
    node[pos=0.43,right,inner sep=2pt] (bobkeeplabel) {$\mathtt B_1$};
  % Shorten the white underlay so it clears the channel border.
  \draw[semithick] (aSend)
    .. controls (3.34,5.35) and (6.66,7.6) .. (6.66,8.25);
  \draw[semithick,preaction={draw,white,line width=7pt,shorten <=5pt}] (bSend)
    .. controls (6.66,5.35) and (3.34,7.6) .. (3.34,8.25);
  \node (alicesendlabel) at (4.25,5.05) {$\mathtt B_2,x$};
  \node (bobsendlabel) at (5.75,5.05) {$\mathtt A_2,y$};
  \draw[wire] (3.34,8.25) -- (aReceive);
  \draw[wire] (6.66,8.25) -- (bReceive);

  % Leave space between the state cut and the channel input arrowheads.
  \draw[densely dotted,thin] (-0.3,3) -- (10.3,3);
  \node[anchor=west] (stateinitial) at (10.45,3)
    {$\rhoinit_{\ol{\mathtt Q}\mathtt{QLR}}$};
  \draw[densely dotted,thin] (-0.3,9) -- (10.3,9);
  \node[anchor=west] (statemid) at (10.45,9)
    {$\rhopre_{\ol{\mathtt Q}\mathtt{AB}}$};
  \draw[densely dotted,thin] (-0.3,12) -- (10.3,12);
  \node[anchor=west] (statefinal) at (10.45,12)
    {$\rhopost_{\ol{\mathtt Q}\mathtt Q_0\mathtt Q_1}$};
  \draw[wire,dashed] (alicepost.north) -- (0,13.9)
    node[pos=0.55,above right,inner sep=1pt] (aliceoutputlabel) {$\mathtt Q_0$};
  \draw[wire,dashed] (bobpost.north) -- (10,13.9)
    node[pos=0.55,above left,inner sep=1pt] (boboutputlabel) {$\mathtt Q_1$};
  \fill (0,13.9) circle (1.2pt);
  \fill (10,13.9) circle (1.2pt);
  % Optional presentation-only highlighting; absent in the paper.
  \ifdefined\protocolstatehighlights
    \protocolstatehighlights
  \fi
\end{tikzpicture}
\caption{Spacetime view of an $f$-routing attack and the evolution of the state. The lower arc denotes pre-shared entanglement, and $V_0$ retains $\ol{\mathtt Q}$. For $b=f(x,y)$, the verifiers test $\mathtt Q_b$.}
\label{fig:main-nlqc}
\end{figure}
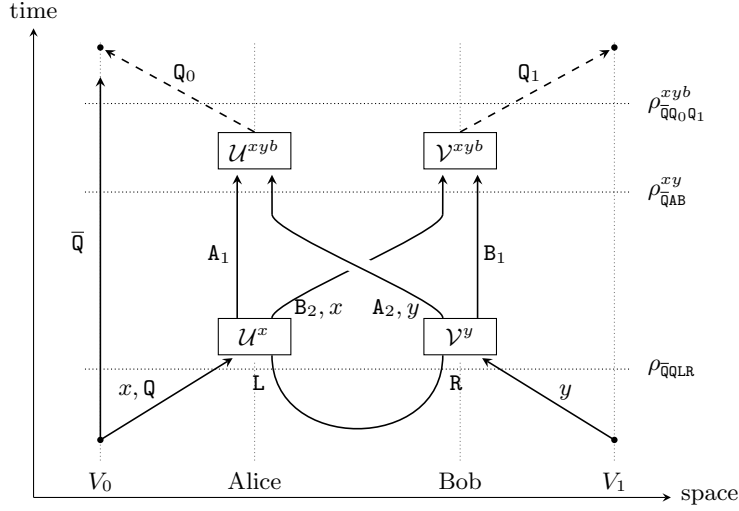

We say that the strategy is $(1-\eps)$-correct on input $(x,y)$ if
\[
    F\ps{(\rhopost)_{\ol{\mathtt{Q}}\mathtt{Q}_{f(x,y)}} \; , \; \Phi_{\ol{\mathtt{Q}}\mathtt{Q}}} \ge 1-\eps.
\]

\section{The Universal Fidelity Gap}\label{sec:fidelity bounds}
Our main theorem in this section is a universal fidelity gap for any $f$-routing strategy that is $(1-\eps)$-correct on an input $(x,y)$. 

\begin{theorem}\label{thm:universal-fidelity-gap}
Let $\eps>0$.  Consider an $f$-routing strategy that is $(1-\eps)$-correct on an input $(x,y)$, and assume $f(x,y)=1$. Let $\rho_{\ol{\mathtt{Q}} \mathtt{AB}}=(\rhopre)_{\ol{\mathtt{Q}} \mathtt{AB}}$ be its mid-protocol state, and define
\begin{align*}
    F_0 &\defeqq F\ps{\rho_{\ol{\mathtt{Q}} \mathtt{A}} \; , \; \rho_{\ol{\mathtt{Q}}}\otimes\rho_{\mathtt{A}}},\\
    F_1 &\defeqq F\ps{\rho_{\ol{\mathtt{Q}} \mathtt{B}} \; , \; \rho_{\ol{\mathtt{Q}}}\otimes\rho_{\mathtt{B}}}.
\end{align*}
Then
\[
    F_0\geq 1-4\eps,
    \qquad
    F_1\leq\frac1d+\sqrt{2\eps}.
\]
\end{theorem}
\begin{proof}
    The two bounds follow from \cref{thm:soundness product state X,thm:completeness product state}, respectively.
\end{proof}

For routing with perfect 1-sided error, namely when $\eps=0$ whenever $f(x,y)=1$, \citet[Lemma~5]{AsadiCulfMay2025} established that $F_0 = 1$ and $F_1 < 1$. Our theorem gives quantitative fidelity bounds that remain valid for nonzero error.

\subsection{Soundness Separability}

\begin{lemma}\label{thm:soundness product state X}
Let $\eps > 0.$
Consider an $f$-routing strategy that is $\ps{1-\eps}$-correct on an input $\ps{x,y}$ satisfying $f\ps{x,y}=1$, and let $\rho_{\ol{\mathtt{Q}} \mathtt{AB}}=\rhopre_{\ol{\mathtt{Q}} \mathtt{AB}}$ be its mid-protocol state. Then
\begin{align*}
    F_0
    \defeqq
    F\ps{\rho_{\ol{\mathtt{Q}} \mathtt{A}} \; , \; \rho_{\ol{\mathtt{Q}}}\otimes\rho_{\mathtt{A}}}
    \ge 1-4\eps.
\end{align*}
\end{lemma}

\begin{proof}
We begin with the triangle inequality that will drive the proof. For a
state $\tau_{\mathtt{A}}$ to be chosen below, define
\begin{align*}
    \alpha 
    &\defeqq
    F\ps{\rho_{\ol{\mathtt{Q}} \mathtt{A}} \; , \; \rho_{\ol{\mathtt{Q}}}\otimes\tau_{\mathtt{A}}},
    &
    \beta
    &\defeqq
    F\ps{
        \rho_{\ol{\mathtt{Q}}}\otimes\tau_{\mathtt{A}} \; , \;
        \rho_{\ol{\mathtt{Q}}}\otimes\rho_{\mathtt{A}}
    }.
\end{align*}
Applying \cref{thm:fidelity triangle} through the intermediate product
state $\rho_{\ol{\mathtt{Q}}}\otimes\tau_{\mathtt{A}}$ gives
\begin{align*}
    F_0
    &=
    F\ps{\rho_{\ol{\mathtt{Q}} \mathtt{A}} \; , \; \rho_{\ol{\mathtt{Q}}}\otimes\rho_{\mathtt{A}}} \\
    &\ge
    \alpha \beta-\sqrt{1-\alpha ^2}\sqrt{1-\beta^2}.
\end{align*}
We will later choose $\tau_{\mathtt{A}}$ such that $\alpha ,\beta\ge 1-\eps$. Substituting
these two bounds into the above will then give the desired lower bound on $F_0$:
\begin{align*}
    F_0
    &\ge
    \ps{1-\eps}^2
    -
    \sqrt{1-\ps{1-\eps}^2}
    \sqrt{1-\ps{1-\eps}^2}
    \\
    &=
    \ps{1-\eps}^2-\ps{2\eps-\eps^2}
    \\
    &=
    1-4\eps+2\eps^2
    \\
    &\ge
    1-4\eps.
\end{align*}
We move forward to choosing $\tau_{\mathtt{A}}$. Since the strategy is $\ps{1-\eps}$-correct and $f\ps{x,y}=1$, Bob's
post-communication decoder $\cV^{xy1}_{\mathtt{B}\to \mathtt{Q}_1}$ produces the state
\begin{align*}
    \sigma_{\ol{\mathtt{Q}} \mathtt{Q}_1\mathtt{A}}
    \defeqq
    \ps{\operatorname{id}_{\ol{\mathtt{Q}} \mathtt{A}}\otimes\cV^{xy1}_{\mathtt{B}\to \mathtt{Q}_1}}
    \ps{\rho_{\ol{\mathtt{Q}} \mathtt{AB}}},
\end{align*}
Since Alice's post-communication channel is trace preserving and acts only on $\mathtt{A}$, applying it does not change the marginal on $\ol{\mathtt{Q}} \mathtt{Q}_1$.  Hence, the correctness condition from \cref{sec:setup} implies that its marginal satisfies
\begin{align*}
    F\ps{\sigma_{\ol{\mathtt{Q}} \mathtt{Q}_1},\Phi_{\ol{\mathtt{Q}} \mathtt{Q}_1}}
    \ge 1-\eps.
\end{align*}
Because $\Phi_{\ol{\mathtt{Q}} \mathtt{Q}_1}$ is pure, Uhlmann's theorem provides a state
$\tau_{\mathtt{A}}$ such that
\begin{align}\label{eq:uhlmann extension}
    F\ps{
        \sigma_{\ol{\mathtt{Q}} \mathtt{Q}_1\mathtt{A}} \; ,\;
        \Phi_{\ol{\mathtt{Q}} \mathtt{Q}_1}\otimes\tau_{\mathtt{A}}
    }
    &=
    F\ps{\sigma_{\ol{\mathtt{Q}} \mathtt{Q}_1} \; ,\; \Phi_{\ol{\mathtt{Q}} \mathtt{Q}_1}}
    \ge 1-\eps.
\end{align}
We first bound $\alpha $. Since $\cV^{xy1}_{\mathtt{B}\to \mathtt{Q}_1}$ is trace-preserving
and acts only on $\mathtt{B}$,
\begin{equation}\label{eq:firsteq}
    \begin{aligned}
        \tr_{\mathtt{Q}_1}\ps{\sigma_{\ol{\mathtt{Q}} \mathtt{Q}_1\mathtt{A}}}
    &=
    \tr_{\mathtt{Q}_1}\!\left[
        \ps{\operatorname{id}_{\ol{\mathtt{Q}} \mathtt{A}}\otimes\cV^{xy1}_{\mathtt{B}\to \mathtt{Q}_1}}
        \ps{\rho_{\ol{\mathtt{Q}} \mathtt{AB}}}
    \right] 
    \\
    &=
    \tr_{\mathtt{B}}\ps{\rho_{\ol{\mathtt{Q}} \mathtt{AB}}}
    \\
    &=
    \rho_{\ol{\mathtt{Q}} \mathtt{A}}.
    \end{aligned}
\end{equation}
Moreover, since the register $\ol{\mathtt{Q}}$ is untouched throughout the protocol,
\[
    \rho_{\ol{\mathtt{Q}}}=\tr_{\mathtt{Q}}\ps{\Phi_{\ol{\mathtt{Q}} \mathtt{Q}}}=\frac{I_{\ol{\mathtt{Q}}}}{d}.
\]
Therefore,
\begin{align*}
    \alpha 
    &=
    F\ps{\rho_{\ol{\mathtt{Q}} \mathtt{A}} \; , \; \rho_{\ol{\mathtt{Q}}}\otimes\tau_{\mathtt{A}}} \\
    &=
    F\ps{
        \tr_{\mathtt{Q}_1}\ps{\sigma_{\ol{\mathtt{Q}} \mathtt{Q}_1\mathtt{A}}} \; , \;
        \rho_{\ol{\mathtt{Q}}} \otimes \tau_{\mathtt{A}}
    }
    && \text{(\cref{eq:firsteq})} \\
    &=
    F\ps{
        \tr_{\mathtt{Q}_1}\ps{\sigma_{\ol{\mathtt{Q}} \mathtt{Q}_1\mathtt{A}}} \; , \;
        \tr_{\mathtt{Q}_1}\ps{\Phi_{\ol{\mathtt{Q}} \mathtt{Q}_1}\otimes\tau_{\mathtt{A}}}
    }
    && \text{($\rho_{\ol{\mathtt{Q}}}=I_{\ol{\mathtt{Q}}}/d$)} \\
    &\ge
    F\ps{
        \sigma_{\ol{\mathtt{Q}} \mathtt{Q}_1\mathtt{A}} \;, \;
        \Phi_{\ol{\mathtt{Q}} \mathtt{Q}_1}\otimes\tau_{\mathtt{A}}
    }
    && \text{(monotonicity under $\tr_{\mathtt{Q}_1}$)}
    \notag\\
    &\ge 1-\eps.
    && \text{(\cref{eq:uhlmann extension})}
\end{align*}
We next bound $\beta$. 
Using multiplicativity of fidelity for tensor product, it follows that
\begin{align*}
    \beta
    &=
    F\ps{
        \rho_{\ol{\mathtt{Q}}}\otimes\tau_{\mathtt{A}} \;, \;
        \rho_{\ol{\mathtt{Q}}}\otimes\rho_{\mathtt{A}}
    }
    \notag\\
    &=
    F\ps{\rho_{\ol{\mathtt{Q}}} \; , \; \rho_{\ol{\mathtt{Q}}}} \cdot F\ps{\tau_{\mathtt{A}} \; , \; \rho_{\mathtt{A}}} \\
    &=
    1 \cdot F\ps{\tau_{\mathtt{A}} \;, \; \rho_{\mathtt{A}}},
\end{align*}
because for any state $\omega$, its fidelity with itself is $F(\omega, \omega) = 1$.

Finally, using the monotonicity of fidelity, 
\begin{align*}
    \beta &= 
    F\ps{\rho_{\mathtt{A}} \; , \; \tau_{\mathtt{A}}} \\
    &=
    F\ps{
        \tr_{\ol{\mathtt{Q}}}\ps{\rho_{\ol{\mathtt{Q}} \mathtt{A}}} \; , \; \tr_{\ol{\mathtt{Q}}} \ps{\rho_{\ol{\mathtt{Q}}} \otimes\tau_{\mathtt{A}}}
    } \\
    &\ge
    F\ps{\rho_{\ol{\mathtt{Q}} \mathtt{A}} \; ,\; \rho_{\ol{\mathtt{Q}}}\otimes\tau_{\mathtt{A}}}
    && \text{(monotonicity under $\tr_{\ol{\mathtt{Q}}}$)} \\
    &=
    \alpha  \\
    &\ge 1-\eps.
\end{align*}
\end{proof}

\subsection{Completeness inseparability}

\begin{lemma}\label{thm:completeness product state}
    Consider an $f$-routing strategy that is $(1-\eps)$-correct on an input $(x,y)$ satisfying $f(x,y)=1$, and let $\rho_{\ol{\mathtt{Q}} \mathtt{AB}}=(\rhopre)_{\ol{\mathtt{Q}} \mathtt{AB}}$ be its mid-protocol state.  Then
    \begin{align*}
    F_1\defeqq F\pss{\rho_{\ol{\mathtt{Q}} \mathtt{B}} \; , \; \rho_{\ol{\mathtt{Q}}} \otimes \rho_{\mathtt{B}}} \leq \frac{1}{d} + \sqrt{2\eps}.
    \end{align*}
\end{lemma}
\begin{proof}
To bound the target fidelity, we apply Bob's decoding channel $\cV^{xy1}_{\mathtt{B}\to \mathtt{Q}_1}$ on both states and observe how they relate to each other. Denote
\begin{align*}
    \sigma_{\ol{\mathtt{Q}} \mathtt{Q}_1} \defeq
    \ps{\operatorname{id}_{\ol{\mathtt{Q}}}\otimes\cV^{xy1}_{\mathtt{B}\to \mathtt{Q}_1}}\ps{\rho_{\ol{\mathtt{Q}} \mathtt{B}}}, \\
    \omega_{\ol{\mathtt{Q}} \mathtt{Q}_1} \defeq
    \ps{\operatorname{id}_{\ol{\mathtt{Q}}}\otimes\cV^{xy1}_{\mathtt{B}\to \mathtt{Q}_1}}
    \ps{\rho_{\ol{\mathtt{Q}}} \otimes \rho_{\mathtt{B}}}.
\end{align*}
By the data-processing inequality for fidelity (\cref{thm:fidelity data processing}):
\begin{align*}
F\ps{\rho_{\ol{\mathtt{Q}} \mathtt{B}} \;, \; \rho_{\ol{\mathtt{Q}}} \otimes \rho_{\mathtt{B}}}
\leq F\ps{\ps{\operatorname{id}_{\ol{\mathtt{Q}}}\otimes\cV^{xy1}_{\mathtt{B}\to \mathtt{Q}_1}}\ps{\rho_{\ol{\mathtt{Q}} \mathtt{B}}} \;, \;  \ps{\operatorname{id}_{\ol{\mathtt{Q}}}\otimes\cV^{xy1}_{\mathtt{B}\to \mathtt{Q}_1}}\ps{\rho_{\ol{\mathtt{Q}}} \otimes \rho_{\mathtt{B}}}}
= F\pss{\sigma, \omega}.
\end{align*}
To bound the fidelity $F(\sigma,\omega)$, we compare the probabilities with which the two states pass the EPR test. Consider the channel $\cM$ that measures in the $\set{\Phi, I - \Phi}$ basis, namely for every state $\tau$ it acts as follows
\[
\mathcal M(\tau) = \tr\ps{\Phi \tau} \ketbraa{0} + \tr\ps{(I - \Phi)\tau}\ketbraa{1}.
\]
Since the outputs of $\cM$ are classical states, their fidelity can be computed directly:
\begin{spliteq}\label{eq:fpb-two-outcome}
F\ps{\sigma \; , \; \omega}
&\leq F\ps{\cM(\sigma) \; , \;\cM(\omega)} \\
&= \tr \sqrt{\cM(\sigma) \cM(\omega)} \\
&= \sqrt{\tr\ps{\Phi \sigma}\tr\ps{\Phi \omega}} + \sqrt{\tr\ps{(I - \Phi)\sigma}\tr\ps{(I - \Phi)\omega}} \\
&\leq \sqrt{1 \cdot \tr(\Phi \omega)} + \sqrt{\tr((I- \Phi) \sigma) \cdot 1} \\
&=\sqrt{\tr(\Phi \omega)} + \sqrt{1 - \tr(\Phi \sigma)}. 
\end{spliteq}
We bound these quantities in \cref{thm:trace phi sigma,thm:trace phi omega} and deduce:

\begin{align*}
F\ps{\rho_{\ol{\mathtt{Q}} \mathtt{B}} \;, \; \rho_{\ol{\mathtt{Q}}} \otimes \rho_{\mathtt{B}}}
\le \frac1{d} + \sqrt{2\eps}.
\end{align*}
\end{proof}

\begin{lemma}\label{thm:trace phi sigma}
    Under the hypotheses of \cref{thm:completeness product state}, let
    \[
        \sigma_{\ol{\mathtt{Q}} \mathtt{Q}_1}
        \defeq
        \ps{\operatorname{id}_{\ol{\mathtt{Q}}}\otimes\cV^{xy1}_{\mathtt{B}\to \mathtt{Q}_1}}
        \ps{\rho_{\ol{\mathtt{Q}} \mathtt{B}}}.
    \]
    Then $\tr(\Phi \sigma) \ge 1-2\eps$.
\end{lemma}
\begin{proof}
Since Alice's post-communication channel is trace preserving and acts only on $\mathtt{A}$, the correctness condition implies
\begin{align*}
F\ps{\sigma \; , \;\Phi_{\ol{\mathtt{Q}} \mathtt{Q}_1}} \geq 1 - \eps.
\end{align*}
Therefore
\begin{align*}
\tr\ps{\Phi \sigma} = F\ps{\sigma,\Phi}^2 \ge (1-\eps)^2 \ge 1-2\eps.
\end{align*}
\end{proof}

\begin{lemma}\label{thm:trace phi omega}
Under the hypotheses of \cref{thm:completeness product state}, let
\[
    \omega_{\ol{\mathtt{Q}} \mathtt{Q}_1}
    \defeq
    \ps{\operatorname{id}_{\ol{\mathtt{Q}}}\otimes\cV^{xy1}_{\mathtt{B}\to \mathtt{Q}_1}}
    \ps{\rho_{\ol{\mathtt{Q}}}\otimes\rho_{\mathtt{B}}}.
\]
Then $\tr\ps{\Phi\omega} = \frac{1}{d^2}$.
\end{lemma}

\begin{proof}
Since $\cV^{xy1}_{\mathtt{B} \to \mathtt{Q}_1}$ acts only on $\mathtt{B}$,
\[
\omega
= \rho_{\ol{\mathtt{Q}}}\otimes\cV^{xy1}_{\mathtt{B} \to \mathtt{Q}_1}\ps{\rho_{\mathtt{B}}}
\defeq \rho_{\ol{\mathtt{Q}}}\otimes\omega_{\mathtt{Q}_1}.
\]
The register $\ol{\mathtt{Q}}$ is the reference half of $\Phi_{\ol{\mathtt{Q}} \mathtt{Q}}$ and is untouched by every map of the protocol, hence its marginal is exactly the maximally mixed state
\[
\rho_{\ol{\mathtt{Q}}} = \tr_{\mathtt{Q}}\ps{\Phi_{\ol{\mathtt{Q}} \mathtt{Q}}} = \frac{I_{\ol{\mathtt{Q}}}}{d}.
\]
Applying \cref{thm:epr measurement} with $A = \rho_{\ol{\mathtt{Q}}}$ and $B = \omega_{\mathtt{Q}_1}$, and using $\tr\ps{\omega_{\mathtt{Q}_1}} = 1$,
\[
\tr\ps{\Phi\omega} = \bra{\Phi} \rho_{\ol{\mathtt{Q}}} \otimes \omega_{\mathtt{Q}_1} \ket{\Phi}
= \frac{1}{d} \tr\ps{\rho_{\ol{\mathtt{Q}}} \omega_{\mathtt{Q}_1}^T}
= \frac{1}{d^2} \tr\ps{\omega_{\mathtt{Q}_1}^T}
= \frac{1}{d^2}.
\]
\end{proof}

\section{Security from Time-Bounded QSZK Hardness}
\newlang{\QSD}{QSD}
\newlang{\coQSD}{coQSD}

\newcommand{\simulator}{\mathsf{Sim}}
\newcommand{\prover}{\mathsf{P}}
\newcommand{\verifier}{\mathsf{Ver}}
\newcommand{\HV}{\mathrm{HV}}

A successful uniform routing strategy produces two efficiently preparable states whose fidelity determines $f(x,y)$.  Converting the fidelity gap into a trace-distance gap gives a reduction to Quantum State Distinguishability, and therefore places $f$ in time-bounded QSZK.

We use the following time-bounded version of Quantum Statistical Zero Knowledge \cite{Watrous02,Watrous09}.
Throughout this section, all time bounds $T\colon\N\to\N$ are time-constructible and satisfy $T(n)\geq n$.

\begin{definition} \label{def:qszk-time}
The complexity class $\QSZK(\cT)$ consists of all languages having a Quantum Statistical Zero-Knowledge proof whose honest verifier $\verifier$ runs in time at most $\cT(n)$ and for which every verifier $\verifier^*$ running in time $S(n)$ has a simulator $\simulator_{\verifier^*}$ running in time $\poly(\cT(n)+S(n))$ on inputs of length $n$. The statistical simulation error is negligible in $n$.
\end{definition}

For a family of Boolean functions
\[
    f = \set{ f_n\colon\powern \times\powern \rightarrow\power{}} ,
\]
we write $f\in\QSZK(\cT)$ when the associated language of $f$ belongs to the class, namely 
\[
L_f \defeqq \sett{x\circ y}{f_n(x,y)=1} \in \QSZK(\cT).
\]
Observe that the usual class $\QSZK$ is recovered by taking the union over polynomially bounded functions $T$.

Our main theorem in this section is a security reduction, from a uniform $f$-routing strategy to a time bounded $\QSZK$ algorithm computing $f$.
\begin{theorem}\label{thm:routing-to-qszk}
There are a universal polynomial $p$ and a constant
$\varepsilon_0>0$ with the following property.

If a $\cT$-uniform $f$-routing strategy is $(1-\varepsilon)$-correct on every input, where $0\leq\varepsilon\leq\varepsilon_0$, then $f\in\QSZK(\cT_{\mathrm{attack}})$, where
\[
    \cT_{\mathrm{attack}}(n) \defeq p\left( \ensuremath{\cT(n)+\log d(n)} \right).
\]
\end{theorem}

We also observe that $\QSZK(\cT)$ can be simulated in deterministic space, which can be read as the quantitative form of $\QIP=\PSPACE$.
Formally, for every time bound $T$, let $\QIP(\cT)$ denote the class of languages having a Quantum Interactive Proof system whose verifier runs in time at most $\cT(n)$.  Then,
\begin{theorem}[\cite{JainJiUpadhyayWatrous11}]\label{thm:qip-time-space}
There is a constant $a\geq 1$ such that, for every time-constructible $\cT$,
\[
    \QIP(\cT) \subseteq \DSPACE\left(\cT^a\right).
\]
\end{theorem}
\begin{proof}
For a (uniform) QIP verifier $V$, denote by $\omega(V,x)$ the maximum acceptance probability of $V$ in an interaction. Consider the following natural complete promise problem for $\QIP$:
\begin{align*}
    L_Y &= \{ (V, x, 1^t) : \text{$V$ runs in time $t$ and $\omega(V,x) \geq 2/3$} \} \\
    L_N &= \{ (V, x, 1^t) : \text{$V$ runs in time $t$ and $\omega(V,x) \leq 1/3$} \}
\end{align*}
Clearly $(L_Y,L_N) \in \QIP$, so there exists $a$ such that $(L_Y,L_N) \in \DSPACE(n^a)$. Any $L \in \QIP$ reduces to $(L_Y,L_N)$ with instance length $O(\cT(n))$, from which it follows that $L \in \DSPACE(\cT^a)$.
\end{proof}

One advantage of our uniform techniques is that the deterministic space hierarchy theorem provides unconditional choices of hard functions.
Indeed, observe that simply by ignoring the zero-knowledge requirement, we have that for some universal constant $a\ge 1$, for every time bound $T$, we have that $\QSZK(\cT)\subseteq\QIP(\cT) \subseteq \DSPACE(\cT^a)$.
Therefore by \cref{thm:routing-to-qszk}:

\begin{corollary}\label{cor:qszk-hardness-security}
There is a family of total Boolean functions $f = \set{f_n \colon \powern \times \powern\to\power{}}_{n \in\N}$ for which no $\cT$-uniform $f$-routing strategy is $(1-\varepsilon)$-correct on every input for any $0\leq\varepsilon\leq\varepsilon_0$, where
\[
    f\in \DSPACE \left(\cT_{\mathrm{attack}}^{a+1}\right) \setminus \QSZK\left(\cT_{\mathrm{attack}}\right).
\]
\end{corollary}

\subsection{Reduction to Quantum State Distinguishability}

For any two density matrices $\rho, \sigma$, we use the distance metric
\[
    \Delta(\rho,\sigma) \defeq \frac12 \norm{\rho - \sigma}_1.
\]

\begin{definition}[Quantum State Distinguishability]
\label{def:resource-bounded-qsd}
Let $\alpha,\beta$ be constants such that $0 \leq \alpha <\beta \leq 1$. The promise problem $\QSD_{\alpha,\beta}$ consists of pairs of quantum circuits $(C_0,C_1)$ having designated output registers of the same dimension. Let $\sigma_b$ be the reduced state on the output register obtained by running $C_b$ on the all-zero state and tracing out the other registers. The promise is:
\begin{enumerate}
    \item \emph{Yes:}
    $\Delta(\sigma_0,\sigma_1)\geq\beta$.
    \item \emph{No:}
    $\Delta(\sigma_0,\sigma_1)\leq\alpha$.
\end{enumerate}
\end{definition}
Similarly, we define the complement problem of $\QSD_{\alpha,\beta}$ as $\coQSD_{\alpha,\beta}$. We use the following result of \cite{Watrous02,Watrous09}.
\begin{theorem}\label{thm:qsd-resource-in-qszkt}
For any constants $0\leq\alpha<\beta^2\leq1$,
\[
    \coQSD_{\alpha,\beta} \in \QSZK.
\]
\end{theorem}

\noindent In what follows, fix an arbitrary input length $n$ and write $q=q(n)\leq\cT(n)$, $d=d(n)$, and $\cT=\cT(n)$.

\begin{lemma}\label{thm:routing-to-qsd}
There are a constant $\varepsilon_0>0$ and a universal polynomial
$p_0$ such that the following holds.

Suppose that an $f$-routing strategy is a $\cT$-uniform strategy that is $(1-\varepsilon)$-correct on
every input, where $0\leq\varepsilon\leq\varepsilon_0$.
Then, for some universal constants $0\leq\alpha<\beta^2\leq1$, there is a deterministic many-one reduction $R$ computable in time $\cT_{\mathrm{reduction}}(n) \defeqq p_0(\cT+\log d)$ from $L_f$ to $\coQSD_{\alpha,\beta}$.
\end{lemma}

\begin{proof}
On input $x\circ y$, where $x,y\in\powern$, the deterministic reduction
first generates the resource preparation and Alice's and Bob's
pre-communication channels,
\[
    \cI^n\gets\genresource(1^n),
    \qquad
    \cU^x\gets\genalice(1^n,x),
    \qquad
    \cV^y\gets\genbob(1^n,y),
\]
and then constructs a circuit preparing the mid-protocol state
\[
    \rhoinit_{\tt{\Qb QLR}}
    =
    \Phi_{\tt{\Qb Q}}\otimes
    \cI^n_{\tt{LR}\to\tt{LR}}\ps{\ketbraa{0^{2q}}_{\tt{LR}}},
\]
\[
    (\rhopre)_{\tt{\Qb AB}}
    =
    \ps{\operatorname{id}_{\tt{\Qb}}\otimes
    \cU^x_{\tt{QL}\to\tt{A_1B_2}}\otimes
    \cV^y_{\tt{R}\to\tt{B_1A_2}}}\ps{\rhoinit}.
\]
By \cref{sec:setup}, the generators run in time at most $\cT$, the generated channels have total description length at most $3\cT$, and each pre-communication channel uses at most $q\leq\cT$ qudits. Additionally, the reduction needs to prepare the $d$-dimensional EPR pair, which takes $\log d$ time. Thus, for a universal polynomial $p_0$, the reduction runs in time at most
\[
    \cT_{\mathrm{reduction}}(n)
    =
    p_0\left(\ensuremath{\cT+\log d} \right).
\]
We can therefore deterministically generate
descriptions of two state-preparation circuits $C_0^{x,y}$ and
$C_1^{x,y}$ as follows:
\begin{enumerate}
    \item The circuit $C_0^{x,y}$ prepares $(\rhopre)_{\tt{\Qb AB}}$ once, discards $\tt{B}$,
    and outputs
    \[
        \sigma^{xy}_0
        =
        \rho^{xy}_{\tt{\Qb A}}.
    \]

    \item The circuit $C_1^{x,y}$ runs two independent copies of
    the same preparation, keeps $\tt{\Qb}$ from the first copy and $\tt{A}$ from the second
    copy, and outputs
    \[
        \sigma^{xy}_1
        =
        \rho^{xy}_{\tt{\Qb}}\otimes\rho^{xy}_{\tt{A}}.
    \]
\end{enumerate}
The two circuit descriptions have total length
$O(\cT_{\mathrm{reduction}}(n))$ and can be generated deterministically in
time $O(\cT_{\mathrm{reduction}}(n))$. Enlarging $p_0$ by a universal
constant factor, if necessary, we may bound both quantities by
$\cT_{\mathrm{reduction}}(n)$.

By the routing fidelity bounds in \cref{sec:fidelity bounds},
using $1/d\leq1/2$,
\begin{align*}
    f(x,y)=1
    &\quad\Longrightarrow\quad
    F\left(\sigma^{xy}_0,\sigma^{xy}_1\right)
        \geq1-4\varepsilon,\\
    f(x,y)=0
    &\quad\Longrightarrow\quad
    F\left(\sigma^{xy}_0,\sigma^{xy}_1\right)
        \leq\frac12+\sqrt{2\varepsilon}.
\end{align*}
Then by Fuchs--van de Graaf inequalities (\cref{lem:fvg}) give
\begin{align*}
    f(x,y)=1
    &\quad\Longrightarrow\quad
    \Delta\left(\sigma^{xy}_0,\sigma^{xy}_1\right)
        \leq
        \sqrt{1-(1-4\varepsilon)^2},\\
    f(x,y)=0
    &\quad\Longrightarrow\quad
    \Delta\left(\sigma^{xy}_0,\sigma^{xy}_1\right)
        \geq
        \frac12-\sqrt{2\varepsilon}.
\end{align*}
We may fix a universal constant $\varepsilon_0>0$ such that
\[
    \alpha
    :=
    \sqrt{1-(1-4\varepsilon_0)^2}
    <
    \left(
        \frac12-\sqrt{2\varepsilon_0}
    \right)^2
    =:
    \beta^2.
\]
Thus, for every $\varepsilon\leq\varepsilon_0$,
\begin{align*}
    f(x,y)=1
    &\quad\Longrightarrow\quad
    \Delta\left(\sigma^{xy}_0,\sigma^{xy}_1\right)
        \leq\alpha,\\
    f(x,y)=0
    &\quad\Longrightarrow\quad
    \Delta\left(\sigma^{xy}_0,\sigma^{xy}_1\right)
        \geq\beta.
\end{align*}
Consequently, the deterministic map $(x,y) \longmapsto \left(1^n,C_0^{x,y},C_1^{x,y}\right)$
maps yes-instances of $L_f$ to yes-instances of
$\coQSD_{\alpha,\beta}$ and no-instances of $L_f$ to no-instances.
\end{proof}

\begin{proof}[Proof of \cref{thm:routing-to-qszk}]
We specify a $\QSZK(\cT)$ verifier $V$ for $L_f$. Given input $(x,y)$, compute $(C_0,C_1) = R(x,y)$ as given by \cref{thm:routing-to-qsd}. Then run the $\QSZK$ verifier for $\coQSD$ given by \cref{thm:qsd-resource-in-qszkt}. The total running time of $V$ is at most polynomial in $\cT_\mathrm{reduction}$; we can therefore choose a polynomial $p$ such that its running time is at most $T_\mathrm{attack}$.
\end{proof}

\section{Security against Clifford+$T$ Adversaries}\label{sec:cliffordT}

\subsection{Main Results}

In this section, we focus on Clifford+$T$ adversaries.

\begin{definition}
    Let $\cT,t\colon\N\to\N$.
    A $\cT$-uniform $f$-routing strategy is a $(\cT,t)$-uniform Clifford+$T$ strategy if the resource generated channels $\cI^n, \cU^x, \cV^y$ are all Clifford+$T$ circuits with at most $t(n)$ magic gates each, for every $x,y\in\powern$.
\end{definition}
For the remainder of this section, fix an arbitrary input length $n$, and write $d=d(n)$, $t=t(n)$, and $\cT=\cT(n)$. For simplicity and ease of notations, we consider circuits using at most $q=q(n)\leq\cT(n)$ qudits.

Fix two Clifford+$T$ circuits $W_\rho$ and $W_\sigma$, each of size at most $\abs{W_{\rho}}, \abs{W_{\sigma}} \le L$ with at most $t$ magic gates, on registers $\mathtt{AE}_\rho$ and $\mathtt{AE}_\sigma$ comprising $q_\rho$ and $q_\sigma$ qubits, respectively, and write
\[
\ket{\rho}_{\mathtt{AE}_\rho}\defeqq W_\rho\ket{0^{q_\rho}},
\qquad
\ket{\sigma}_{\mathtt{AE}_\sigma}\defeqq W_\sigma\ket{0^{q_\sigma}},
\]
\[
\rho_{\mathtt{A}}\defeqq\tr_{\mathtt{E}_\rho}\ketbraa{\rho},
\qquad
\sigma_{\mathtt{A}}\defeqq\tr_{\mathtt{E}_\sigma}\ketbraa{\sigma}.
\]
The following claim bounds the cost of preparing a product of marginals of states, each of which is prepared by a Clifford+$T$ circuit.
\begin{claim}\label{lem:cliffordT-routing-preparations}
    If $\rho_{\tt{\Qb A}}$ is a reduced state of a Clifford+$T$ circuit on $q$ qubits of size at most $L$ with at most $t$ magic gates, then $\rho_{\tt{\Qb}}\otimes\rho_{\tt{A}}$ is a reduced state of a Clifford+$T$ circuit on $2q$ qubits of size at most $2L$ with at most $2t$ magic gates.
\end{claim}
\begin{proof}
    Let $W$ be the given circuit and, abusing notation, write $\ket{\rho}_{\tt{\Qb AE}}=W\ket{0^q}$ for a purification of $\rho_{\tt{\Qb A}}$. The circuit $W\otimes W$ prepares $\ket{\rho}_{\tt{\Qb AE}}\otimes\ket{\rho}_{\tt{\Qb' A' E'}}$, and tracing out all registers except $\tt{\Qb A'}$ gives $\rho_{\tt{\Qb}}\otimes\rho_{\tt{A'}}$, yielding the desired state. The circuit acts on $2q$ qubits, has size at most $2L$, and contains at most $2t$ magic gates.
\end{proof}

Our main result approximates the fidelity of reduced states prepared by Clifford+$T$ circuits:

\begin{theorem}\label{thm:cliffordT-state-fidelity}
    For every $\delta\in(0,1)$, given $W_\rho$ and $W_\sigma$, a deterministic classical algorithm outputs $\wt{F}$ satisfying
    \[
    \abs{\wt{F}-F(\rho_{\mathtt{A}},\sigma_{\mathtt{A}})}\leq\delta.
    \]
    The algorithm runs in deterministic time $\poly(2^{t}, L, q_\rho+q_\sigma,\log(1/\delta))$.
\end{theorem}
The proof is developed in \cref{sec:cliffordT-residual}.
Combining \cref{lem:cliffordT-routing-preparations,thm:cliffordT-state-fidelity,thm:universal-fidelity-gap} gives the following corollary.

\begin{corollary}\label{thm:cliffordT-security}
    There are universal constants $c_0$ and $\eps_0>0$ such that the following holds. Suppose $d$ is a power of $2$.\footnote{The restriction is for simplicity.} Let $f=\set{f_n\colon\powern\times\powern\to\power{}}$ be a family of functions.

    If there exists a $(\cT,t)$-uniform Clifford+$T$ $f$-routing strategy that is $(1-\eps)$-correct on every input at every input length for some $0\leq\eps\leq\eps_0$, then
    \[
    f\in\DTIME\ps{\ensuremath{(\cT d 2^t)^{c_0}}}.
    \]
\end{corollary}
\begin{proof}
    Since $d$ is a power of two, the challenge EPR pair is prepared using $O(\log d)$ Clifford gates and no magic gates. Composing this preparation with the three adversarial circuits gives a mid-protocol purification on $O(q\log d)$ qubits of circuit size $O(\cT+\log d)$ with at most $3t$ magic gates. Applying \cref{lem:cliffordT-routing-preparations,thm:cliffordT-state-fidelity} and the constant fidelity gap from \cref{thm:universal-fidelity-gap} gives the claimed running time, including the strategy-generation cost $\cT$.
\end{proof}

\subsection{Preliminaries}\label{sec:cliffordT-overview}
Recall that the \textbf{stabilizer group} of a nonzero pure vector (a possibly unnormalized pure state) $\ket{\psi}$ on $q$ qubits is
\[
\stabgroup(\psi)
\defeq
\sett{P\in \Pauli_q}{P\ket{\psi}=\ket{\psi}}.
\]
The rank of $\stabgroup(\psi)$ is the number of its independent generators. A \textbf{stabilizer state} is a pure state whose stabilizer group has rank $q$, namely, whose stabilizer group is generated by $q$ independent Pauli operators.
The \textbf{stabilizer nullity} \cite{BeverlandEtAl2020} counts the number of independent stabilizers missing from a full set of $q$ generators.

In this section, we also work with ``pure vectors'', namely non-normalized pure states.
\begin{definition}[\cite{BeverlandEtAl2020}]
    The stabilizer nullity of a $q$-qubit pure vector $\ket \psi \neq 0$ is defined as
    \[
    \nu(\psi)
    \defeq
    q - \rank\bigl(\stabgroup(\psi)\bigr).
    \]
\end{definition}

We make use of stabilizer decompositions.
\begin{definition}[{\cite{BravyiEtAl2019}}]
    An $m$-stabilizer decomposition of a pure vector $\ket{\psi}$ is a linear combination of $m$ many stabilizer states $\ket{s_1},\ldots,\ket{s_m}$:
    \[
    \ket{\psi}=\sum_{j=1}^{m}\alpha_j\ket{s_j}.
    \]
    The decomposition is \emph{explicit}\footnote{For the Clifford+$T$ constructions below, coefficients lie in $\mathbb Q(\sqrt{2},i)$, since $e^{i\pi/4}=(1+i)/\sqrt{2}$. We encode their real and imaginary parts separately as $a+b\sqrt{2}$ and $c+d\sqrt{2}$, with rational $a,b,c,d$, and include their bit lengths in the input description length.} if the coefficients $\alpha_j$ and Clifford circuits $D_j$ preparing each $\ket{s_j}=D_j\ket{0^q}$ are given explicitly.
\end{definition}

We combine the expansion with a set of independent stabilizer generators.

\begin{definition}\label{def:mnu-stabilizer-decomposition}
    Fix integers $q$, $m \ge 1$ and $\nu \in [0,q]$.
    An $(m,\nu)$-stabilizer decomposition of a $q$-qubit vector $\ket{\psi}$ consists of:
    \begin{enumerate}
        \item an $m$-stabilizer decomposition $\ket{\psi}=\sum_{j=1}^m\alpha_j\ket{s_j}$;

        \item $q-\nu$ independent commuting Hermitian Pauli operators $P_1,\ldots,P_{q-\nu}$ satisfying the following:
        \begin{itemize}
            \item $-I\notin\langle P_1,\ldots,P_{q-\nu}\rangle$,
            \item $P_j\ket{\psi}=\ket{\psi}$ for all $j\in[q-\nu]$.
        \end{itemize}
    \end{enumerate}
    The decomposition is \emph{explicit} if its $m$-stabilizer decomposition is explicit and all $P_j$'s are given explicitly.
\end{definition}

We use the following computational tool for partial projections of stabilizer states.

% \oren{ to verify}
%
\begin{theorem}[See {\cite[Sections~2.2 and~4.1]{BravyiEtAl2019}}]\label{thm:stabilizer-computation}
    There exists a deterministic algorithm that, given two Clifford circuits $A,B$, preparing the stabilizer states $\ket a = A\ket{0^p}$ and $\ket b = B\ket{0^q}$, where $p\leq q$, and given a $p$-qubit subsystem $\mathtt S$ of $\ket b$, outputs a Clifford circuit $D$ and scalar $\alpha \in \C$ such that
    \begin{align*}
        \alpha D \ket {0^{q-p}} = (\bra a_{\mathtt S}\otimes I_{\bar{\mathtt S}})\ket b.
    \end{align*}
    The algorithm runs in deterministic time $\poly(q,p, \abs A, \abs B)$.
\end{theorem}

We make use of the following formula to compute the fidelity.

\begin{theorem}[{\cite[Lemma~3.21]{Watrous18TQI}}]\label{thm:purification-trace-norm}
    Let $\ket{\rho}_{\mathtt{AE}_\rho}$ and $\ket{\sigma}_{\mathtt{AE}_\sigma}$ be
    purifications of $\rho_{\mathtt{A}}$ and $\sigma_{\mathtt{A}}$. Then
    \[
    \left\|\sqrt{\rho_{\mathtt{A}}}\sqrt{\sigma_{\mathtt{A}}}\right\|_1
    = \left\| \tr_{\mathtt{A}} \left( \ket{\sigma}\!\bra{\rho} \right) \right\|_1.
    \]
\end{theorem}
Therefore, defining
\begin{align*}
    M &\defeqq \tr_{\mathtt{A}}\ps{\ketbra{\sigma}{\rho}},
\end{align*}
we have
\[
F(\rho_{\mathtt{A}},\sigma_{\mathtt{A}})=\norm M_1.
\]
Thus, it suffices to compute and sum the singular values of $M$. Choose computational bases $\{\ket{a}\}$ of $\mathtt{E}_\rho$ and $\{\ket{b}\}$ of $\mathtt{E}_\sigma$, and define the ``vectorization'' of $M$ as
\begin{align}\label{eq:chi-def}
    \ket{\chi}_{\mathtt{E}_\sigma \mathtt{E}_\rho}
    &\defeqq
    \sum_{a,b}\langle b|M|a\rangle\ket{b}_{\mathtt{E}_\sigma}\ket{a}_{\mathtt{E}_\rho}.
\end{align}
Now observe that the Schmidt coefficients of $\ket{\chi}$ are the singular values of $M$, and their sum is the target fidelity.

The algorithm has two parts. First, we give an explicit $(m,\nu)$-stabilizer decomposition for $\ket{\chi}$. Second, we use the decomposition to compute the fidelity.

\subsection{From Clifford+$T$ circuits to $(m,\nu)$-Stabilizer Decompositions}\label{sec:cliffordT-deferred}

Our main lemma in this subsection is the following description of $\ket{\chi}$.

\begin{lemma}\label{lem:bell-contraction-data}
    There exists a deterministic algorithm that, given $W_\rho$ and $W_\sigma$, computes an explicit $(m_\chi,\nu_\chi)$-stabilizer decomposition of $\ket{\chi}$ with
    \[
    m_\chi\leq2^{2t},\qquad \nu_\chi\leq2t.
    \]
    The algorithm runs in time $\poly(2^{t}, L,q_\rho+q_\sigma)$.
\end{lemma}

To this end, we first construct $(m,\nu)$-stabilizer decompositions for Clifford+$T$ states. This construction follows from \cite{BravyiEtAl2019,GuOlivieroLeone2025}; we include a proof for completeness.

\begin{lemma}\label{lem:cliffordT-nullity}
    There exists a deterministic algorithm that, given a Clifford+$T$ circuit $U$ on $q$ qubits of size at most $\abs U \le L$ with at most $t$ magic gates, computes an explicit $(m,\nu)$-stabilizer decomposition of $\ket{\psi}=U\ket{0^q}$, where
    \[
    m\leq2^{t},\qquad \nu\leq\min\set{t,q}.
    \]
    The algorithm runs in deterministic time $\poly(2^{t},L,q)$.
\end{lemma}

\begin{proof}
    We induct on the number $j\leq t$ of magic gates. For $j=0$, $U$ is Clifford, so its output $U\ket {0^q}$ has an explicit $(1,0)$-stabilizer decomposition. Indeed, the stabilizer decomposition is simply $U\ket{0^q}$ as a stabilizer state; the independent stabilizer group generators are  $UZ_{k}U^\dagger$, for all $k\in[q]$.

    Assume the induction hypothesis for $j-1$; we prove the claim for $j$. A final Clifford layer preserves both parameters, so it suffices to consider $U = T_i D V$ where $D$ is a Clifford circuit and $V$ is a Clifford+$T$ circuit with $j-1$ magic gates.  By induction, $\ket v=V\ket{0^q}$ has an explicit $(m,\nu)$-stabilizer decomposition with $m\leq2^{j-1}$ and $\nu\leq\min\set{j-1,q}$.

    Writing the stabilizer decomposition as $\ket v=\sum_{k=1}^m a_{k}\ket{s_{k}}$ and expanding the magic gate as $T_i=\alpha I+\beta Z_i$ (where $\alpha,\beta\in \C$ are scalars), we obtain
    \begin{align*}
        U\ket{0^q}
        =T_iD\ket v=\alpha D\ket v+\beta Z_iD\ket v
        =\sum_{k=1}^m a_{k}\ps{\alpha D\ket{s_{k}}+\beta Z_iD\ket{s_{k}}}.
    \end{align*}
    Since $D$ and $Z_iD$ are Clifford, this gives an explicit stabilizer decomposition with $m'\leq2m\leq2^{j}$ terms.

    Let $\cS$ be the supplied stabilizer subgroup of $\ket v$, of rank $q-\nu$, and define
    \begin{align*}
        \cS'&\defeq\sett{P\in D \cS D^\dagger}{[P,Z_i]=0}.
    \end{align*}
    Now observe that every $P\in \cS'$ stabilizes $D\ket v$. Since $[P,Z_i]=0$ and $T_i=\alpha I+\beta Z_i$, $P$ also commutes with $T_i$. Therefore every $P\in\cS'$ satisfies
    \begin{align*}
        P U\ket{0^q} &=P T_iD\ket v=T_iP D\ket v=U\ket{0^q}.
    \end{align*}
    Conjugation by $D$ preserves the rank of $\cS$. To obtain $\cS'$, we restrict $D\cS D^\dagger$ to the elements commuting with $Z_i$, imposing at most one independent linear constraint. Hence,
    \begin{align*}
        \rank(\cS')&\geq\rank(D\cS D^\dagger)-1=\rank(\cS)-1=q-\nu-1.
    \end{align*}
    We conclude by verifying the claimed running time. The generator updates use Clifford conjugation and Gaussian elimination and take polynomial time in $L,q$. The expansion has at most $2^{j}$ terms, each specified by a Clifford circuit of size $O(L)$ and a product of at most $j$ coefficients $\alpha,\beta$. These products have exact encodings of bit length $O(j+1)$, and their arithmetic takes polynomial time. Thus the total running time is $\poly(2^{t},L,q)$.
\end{proof}

We can now prove:
\begin{proof}[Proof of \cref{lem:bell-contraction-data}]
    We first construct the $m$-stabilizer expansion of $\ket{\chi}$.
    Applying \cref{lem:cliffordT-nullity} to $W_\rho$ and $W_\sigma$ computes the stabilizer decompositions
    \begin{align*}
        \ket{\rho}_{\mathtt{AE}_\rho}
        &=\sum_{i=1}^{m_\rho}\alpha_i^\rho\ket{s_i^\rho}_{\mathtt{AE}_\rho},\qquad
        \ket{\sigma}_{\mathtt{AE}_\sigma} =\sum_{j=1}^{m_\sigma}\alpha_j^\sigma\ket{s_j^\sigma}_{\mathtt{AE}_\sigma},
    \end{align*}
    where each $\ket{s_i^\rho}$ and $\ket{s_j^\sigma}$ is a stabilizer state, and $m_\rho,m_\sigma\leq2^{t}$.

    Substituting these decompositions into \cref{eq:chi-def} gives
    \begin{align*}
        \ket{\chi}
        &=\sum_{i=1}^{m_\rho}\sum_{j=1}^{m_\sigma}
        \underbrace{ \overline{\alpha_i^\rho} \alpha_j^\sigma }_{\alpha_{i,j}}
        \underbrace{ \ps{\sum_u\bra{u}_{\mathtt{A}}\bra{u}_{\mathtt{A}'}}
        \ps{\ket{s_j^\sigma}_{\mathtt{AE}_\sigma}\otimes\overline{\ket{s_i^\rho}}_{\mathtt{A}'\mathtt{E}_\rho}} }_{\ket{s_{i,j}}} \\
        &\defeq \sum_{(i,j) \in [m_{\rho}] \times [m_{\sigma}]} \alpha_{i,j} \ket{s_{i,j}}.
    \end{align*}
    Let $\abs{\tt A}$ be the number of qubits in $\mathtt{A}$, so $q_\chi=q_\rho+q_\sigma-2\abs{\tt A}$. By \cref{thm:stabilizer-computation}, each (non-normalized) vector $\ket{s_{i,j}}$ is a computable scalar times a normalized stabilizer state. The contraction uses the unnormalized Bell vector, so its scalar includes the factor $2^{\abs{\tt A}/2}$. Absorb each scalar into $\alpha_{i,j}$ and henceforth use $\ket{s_{i,j}}$ for the normalized state. For a zero contraction, set $\alpha_{i,j}=0$ and $\ket{s_{i,j}}=\ket{0^{q_\chi}}$. This gives an explicit expansion with at most $m_\chi\leq m_\rho m_\sigma\leq2^{2t}$ terms.

    Next, we compute $q_\chi-\nu_\chi$ independent stabilizers of $\ket{\chi}$. The application of \cref{lem:cliffordT-nullity} above also supplies independent stabilizer generators $P_1,\ldots,P_{q_\rho-\nu_\rho}$ and $Q_1,\ldots,Q_{q_\sigma-\nu_\sigma}$ of $\ket{\rho}$ and $\ket{\sigma}$, respectively, where $\nu_\rho,\nu_\sigma \le t$.
    Thus the product state $\ket{\sigma}\otimes\overline{\ket{\rho}}$ is stabilized by
    \[
    \sett{Q_j\otimes I}{j\in[q_\sigma-\nu_\sigma]}
    \cup
    \sett{I\otimes\overline{P_i}}{i\in[q_\rho-\nu_\rho]}.
    \]
    These generators are jointly independent because the two independent sets act on disjoint registers. The product $\ket{\sigma}\otimes\overline{\ket{\rho}}$ therefore has $q_\rho+q_\sigma$ qubits and $r$ supplied independent stabilizers, where
    \[
    r \ge (q_\rho-\nu_\rho)+(q_\sigma-\nu_\sigma)
    \geq q_\rho+q_\sigma-2t.
    \]
    Recalling that
    \[
    \ket\chi
    =\ps{\sum_u\bra u_{\mathtt A}\bra u_{\mathtt A'}}
    \ps{\ket{\sigma}_{\mathtt{AE}_\sigma}\otimes\overline{\ket{\rho}}_{\mathtt{A'E}_\rho}},
    \]
    observe that $\ket\chi$ is obtained by contracting each of the $\abs{\tt A}$ corresponding pairs of qubits in $\mathtt A,\mathtt A'$ with the unnormalized EPR state. Each Bell contraction can be implemented (up to a scalar) by a two-qubit Clifford unitary followed by two single-qubit projections onto $\ket{0}$.\footnote{Recall that $\ket{\Phi^+}=\mathrm{CNOT}(H\otimes I)\ket{00}$.} Clifford unitaries preserve the number of independent stabilizers, whereas in \cref{lem:stabilizers-one-qubit-projection} we show that each single-qubit projection loses at most one per independent stabilizer. Hence the $\abs{\tt A}$ Bell contractions lose at most $2\abs{\tt A}$ independent stabilizers. The number $q_\chi-\nu_\chi$ of independent generators we obtain therefore satisfies
    \[
    q_\chi-\nu_\chi\geq r-2\abs{\tt A} \geq (q_\rho+q_\sigma-2t) - 2\abs{\tt A} = q_\chi-2t.
    \]
    Thus $\nu_\chi\leq2t$, as required. The generators stabilizing $\ket\chi$ can be computed in polynomial time by applying the procedure in \cref{lem:stabilizers-one-qubit-projection} to each Bell contraction, and computing the coefficients $\alpha_{i,j}$ and stabilizer states $\ket{s_{i,j}}$ takes time $\poly(2^{t},L,q_\rho+q_\sigma)$.
\end{proof}

Finally, we prove the following technical lemma.

\begin{lemma}\label{lem:stabilizers-one-qubit-projection}
    Let $\ket{\phi}_{\mathtt{aE}}$ be a $q$-qubit vector, where $\mathtt{a}$ is a single-qubit register. Consider the post-selected state $\ket {\phi'}$, where $\tt{a}$ is projected onto $\ket0$, and suppose $\ket{\phi'} \neq 0$. That is,
    \[
    \ket{\phi'}_{\mathtt{E}}\defeqq(\bra{0}_{\mathtt{a}}\otimes I_{\mathtt{E}})\ket{\phi}_{\mathtt{aE}} \neq 0.
    \]
    There exists a deterministic algorithm that, given the subgroup of stabilizers $\cS_{\mathrm{in}}\subseteq\stabgroup(\phi)$, outputs a subgroup $\cS_{\mathrm{out}} \subseteq \stabgroup(\phi')$ such that
    \[
    \rank(\cS_{\mathrm{out}})\geq\rank(\cS_{\mathrm{in}})-1.
    \]
    The algorithm runs in deterministic time $\poly(q)$.
\end{lemma}

\begin{proof}
    The Pauli stabilizers of $\ket{0}_{\mathtt{a}}$ are $\set{I_{\mathtt{a}}, Z_{\mathtt{a}}}$. Therefore we argue that the stabilizers of $\ket{\phi'}_{\tt E}$ are as follows:
    \begin{align*}
        \cS_{\mathrm{out}}
        & \defeqq\sett{R}{I_{\mathtt{a}}\otimes R\in \cS_{\mathrm{in}}\text{ or }Z_{\mathtt{a}}\otimes R\in \cS_{\mathrm{in}}}.
    \end{align*}
    To see that, observe that if $Z_{\mathtt{a}}\otimes R\in \cS_{\mathrm{in}}$, then
    \begin{align*}
        R\ket{\phi'}
        &=(\bra{0}_{\mathtt{a}}\otimes R)\ket{\phi} \\
        &=(\bra{0}_{\mathtt{a}}\otimes I_{\mathtt{E}})(Z_{\mathtt{a}}\otimes R)\ket{\phi} \\
        &=(\bra{0}_{\mathtt{a}}\otimes I_{\mathtt{E}})\ket{\phi} \\
        &=\ket{\phi'}.
    \end{align*}
    The same chain holds with $I_{\mathtt{a}}$ in place of $Z_{\mathtt{a}}$, proving that every $R\in \cS_{\mathrm{out}}$ stabilizes $\ket{\phi'}$.

    Computing generators of $\cS_{\mathrm{out}}$ can be found from the given generators of $\cS_{\mathrm{in}}$ by Gaussian elimination in time $\poly(q)$.

    For the rank bound, let $r=\rank(\cS_{\mathrm{in}})$ and let $g_1,\ldots,g_r$ be independent generators of $\cS_{\mathrm{in}}$.
    The construction has two steps: (1) retain the stabilizers commuting with $Z_{\mathtt a}$, and (2) delete their factor on $\mathtt a$. Each step loses at most one independent generator. We show that they cannot both cause a loss, giving $\rank(\cS_{\mathrm{out}}) \geq r-1$.

    \begin{enumerate}[(i)]
        \item If $Z_{\mathtt{a}}\in S_{\mathrm{in}}$, all input stabilizers commute with it, so step (1) loses nothing. We may choose the generating set so that
        \[
        g_1 =Z_{\mathtt{a}}\otimes I_{\mathtt{E}},
        \]
        and write the remaining generators, for suitable $b_j\in\set{0,1}$, as
        \begin{align*}
            g_j &= Z_{\mathtt{a}}^{b_j}\otimes R_j,\qquad j=2,\ldots,r,
        \end{align*}
        Replacing $g_j$ by $g_1^{b_j}g_j=I_{\mathtt{a}}\otimes R_j$ preserves independence. Step (2) therefore leaves the $r-1$ independent generators $R_2,\ldots,R_r$, losing only $g_1$.

        \item If $Z_{\mathtt{a}}\notin S_{\mathrm{in}}$, consider two cases:
        \begin{enumerate}[(a)]
            \item If all $g_j$ commute with $Z_{\mathtt{a}}$, write $g_j=Z_{\mathtt{a}}^{b_j}\otimes R_j$, for some Pauli operators $R_j$ and $b_j\in\set{0,1}$. After deleting the qubit, the $R_j$ remain independent: otherwise, for some nonempty $J\subseteq[r]$,
            \begin{align*}
                \prod_{j\in J}R_j=I_{\mathtt{E}},
            \end{align*}
            which holds if and only if
            \[
            \prod_{j\in J}g_j\in\set{I,Z_{\mathtt{a}}}.
            \]
            However, the product $\prod_{j\in J}g_j$ cannot equal $I$, because the $g_j$ are independent and $J$ is nonempty. It cannot equal $Z_{\mathtt{a}}$ either, because this product belongs to $\cS_{\mathrm{in}}$, whereas $Z_{\mathtt{a}}\notin \cS_{\mathrm{in}}$. Thus $R_1,\ldots,R_r$ remain independent, and all $r$ generators survive.

            \item Otherwise, suppose without loss of generality that $g_1$ anticommutes with $Z_{\mathtt{a}}$. Define $h_1 =g_1$, and for $j\geq2$ define
            \begin{align*}
                h_j &=\begin{cases}
                    g_j & \text{if }[g_j, Z_{\mathtt{a}}] = 0,\\
                    g_1g_j & \text{otherwise}.
                \end{cases}
            \end{align*}
            Clearly $h_1,\ldots, h_r$ are still independent because all the modifications are reversible. Each $h_j$ commutes with $Z_{\mathtt{a}}$, so discarding $g_1$ leaves $r-1$ independent commuting generators.
            Now apply the argument of case (a) to $h_2,\ldots,h_r$ in place of $g_1,\ldots,g_r$. These are independent generators in $\cS_{\mathrm{in}}$ commuting with $Z_{\mathtt{a}}$, and $Z_{\mathtt{a}}\notin \cS_{\mathrm{in}}$ still holds. Therefore, deleting their factor on $\mathtt a$ preserves independence, leaving $r-1$ generators.
        \end{enumerate}
    \end{enumerate}
    Thus in both cases,
    \[
    \rank(\cS_{\mathrm{out}})\geq\rank(\cS_{\mathrm{in}})-1.
    \]
\end{proof}

\subsection{Fidelity from $(m,\nu)$-Stabilizer Decompositions}\label{sec:cliffordT-residual}
We now use the computed $(m,\nu)$-stabilizer decomposition to reduce the fidelity calculation to a small matrix.

The first lemma shows how to locally modify $(m,\nu)$-stabilizer states into another state which is fixed but on at most $2\nu$ qubits.
This local Clifford normal form follows from \cite[proof of Theorem~8]{GuOlivieroLeone2025}; we include a proof for completeness.
\begin{lemma}\label{lem:low-nullity-form}
    Let $\ket{\psi}_{\mathtt{AB}}$ be a $q$-qubit pure vector given by an $(m,\nu)$-stabilizer decomposition.

    Then there exist local Clifford unitaries $D_{\mathtt{A}}, D_{\mathtt{B}}$ and a pure vector $\ket{\gamma}$ on at most $2\nu$ qubits such that
    \begin{align}\label{eq:gamma formalization}
        (D_{\mathtt{A}}\otimes D_{\mathtt{B}})\ket{\psi}_{\mathtt{AB}}
        =
        \ket{0}^{\otimes r}
        \otimes
        \ket{\Phi^+}_{\mathtt{AB}}^{\otimes p}
        \otimes
        \ket{\gamma}.
    \end{align}
    If the decomposition is explicit, then $D_{\mathtt{A}}, D_{\mathtt{B}}$ and the locations of the registers on the right-hand side can be computed in polynomial time.
\end{lemma}
We defer the proof to \cref{sec:clifford-t-gamma thm proof}.

We next compute the amplitudes of $\ket{\gamma}$ from \cref{eq:gamma formalization}.
\begin{lemma}\label{lem:stabilizer-residual-amplitudes}
    Let $\ket{\psi}_{\mathtt{AB}}$ be a vector given by an explicit $(m,\nu)$-stabilizer decomposition. Then there is a deterministic algorithm that computes the computational-basis amplitudes of $\ket{\gamma}$ from \cref{eq:gamma formalization}, exactly in deterministic time $\poly(2^{2\nu}, \ell)$, where $\ell$ is the input description length.
\end{lemma}

\begin{proof}
    Write the $m$-stabilizer decomposition of $\ket \psi$ as $\ket{\psi}=\sum_{j=1}^m\alpha_j\ket{s_j}$. For every computational-basis string $z$ on the residual qubits, \cref{eq:gamma formalization} gives
    \[
    \langle z|\gamma\rangle
    =
    \sum_{j=1}^{m}\alpha_j
    \ps{\bra{0}^{\otimes r}\otimes\bra{\Phi^+}^{\otimes p}\otimes\bra{z}}
    (D_{\mathtt A}\otimes D_{\mathtt B})\ket{s_j}.
    \]
    The coefficients and Clifford circuits are explicit, so each summand, including its phase, is computable exactly by \cref{thm:stabilizer-computation}. Summing the $m$ terms takes polynomial time in the input description length. There are at most $2^{2\nu}$ strings $z$, giving the claimed running time.
\end{proof}

We finish the proof of \cref{thm:cliffordT-state-fidelity} by assembling these steps and computing the Schmidt coefficients of $\ket \chi$.

\begin{proof}[Proof of \cref{thm:cliffordT-state-fidelity}]
    Apply \cref{lem:bell-contraction-data} to obtain an explicit $(m_\chi,\nu_\chi)$-stabilizer decomposition of $\ket{\chi}$. Then \cref{lem:low-nullity-form} computes local Clifford unitaries $D^{(\sigma)}_{\mathtt{E}_\sigma}$ and $D^{(\rho)}_{\mathtt{E}_\rho}$ such that
    \[
    (D^{(\sigma)}_{\mathtt{E}_\sigma}\otimes D^{(\rho)}_{\mathtt{E}_\rho})\ket{\chi}_{\mathtt{E}_\sigma \mathtt{E}_\rho} = \ket{0}^{\otimes r} \otimes \ket{\Phi^+}_{\mathtt{E}_\sigma \mathtt{E}_\rho}^{\otimes p} \otimes \ket{\gamma}.
    \]
    Apply \cref{lem:stabilizer-residual-amplitudes} to compute the at most $2^{2\nu_\chi}$ amplitudes of $\ket{\gamma}$.

    Let $G$ be the coefficient matrix of $\ket{\gamma}$ across the residual qubits of $\mathtt{E}_\sigma$ and $\mathtt{E}_\rho$. Writing $u,v$ for basis strings on these residual registers,
    \[
    G_{u,v}\defeqq\braket{u,v}{\gamma}.
    \]
    The singular values of $G$ are the Schmidt coefficients of $\ket{\gamma}$. Moreover, the local Clifford unitaries and the local $\ket{0}$ factors do not change the Schmidt coefficients, while each Bell pair multiplies their sum by $\sqrt{2}$. Therefore,
    \[
    F(\rho_{\mathtt{A}},\sigma_{\mathtt{A}})
    =
    2^{p/2}\norm{G}_1.
    \]
    Applying \cref{thm:svd} to a sufficiently accurate rational approximation of $2^{p/2}G$ gives the required $\delta$-additive estimate. The bounds on the matrix dimensions and coefficient bit lengths give a running time of
    \[
    \poly(2^{t} , L,q_\rho+q_\sigma,\log(1/\delta)).
    \]
\end{proof}

\subsection{Proof of \cref{lem:low-nullity-form}}\label{sec:clifford-t-gamma thm proof}
We use the following two stabilizer normal-form theorems. The following follows from the original work of \cite{AG04}, though we use the following formalization.
\begin{theorem}[{\cite[Lemma~3.5]{YoganathanJozsaStrelchuk2019}}]\label{thm:clifford-reduction}
    Let $P_1,\ldots,P_k$ be independent commuting Hermitian Pauli operators on $q$ qubits such that $-I\notin\langle P_1,\ldots,P_k\rangle$. Then, there is a Clifford unitary $U$, computable in deterministic time $\poly(q)$, such that
    \[
    \forall i\in[k], \qquad
    U P_i U^\dagger = Z_i.
    \]
\end{theorem}

\begin{theorem}[{\cite[Theorem~7]{LooiGriffiths2011}}]
    \label{thm:tripartite-stabilizer}
    Let $\ket{\Omega}_{\mathtt{ABR}}$ be a pure stabilizer state over $q$ qubits. There exist Clifford unitaries $U_{\mathtt{A}},U_{\mathtt{B}},U_{\mathtt{R}}$,\footnotemark acting on $\mathtt{A},\mathtt{B},\mathtt{R}$, respectively, such that for some integers $m_{\mathtt{ABR}}, m_{\mathtt{AB}}, m_{\mathtt{AR}}, m_{\mathtt{BR}}, r_{\mathtt{A}}, r_{\mathtt{B}}, r_{\mathtt{R}}$:
    \begin{align*}
        (U_{\mathtt{A}}\otimes U_{\mathtt{B}}\otimes U_{\mathtt{R}})\ket{\Omega}_{\mathtt{ABR}}
        =&
        \ket{\mathrm{GHZ}}_{\mathtt{ABR}}^{\otimes m_{\mathtt{ABR}}} \\
        &\otimes
        \ket{\Phi^+}_{\mathtt{AB}}^{\otimes m_{\mathtt{AB}}}
        \otimes
        \ket{\Phi^+}_{\mathtt{AR}}^{\otimes m_{\mathtt{AR}}}
        \otimes
        \ket{\Phi^+}_{\mathtt{BR}}^{\otimes m_{\mathtt{BR}}} \\
        &\otimes \ket{0}_{\mathtt{A}}^{\otimes r_{\mathtt{A}}} \otimes \ket{0}_{\mathtt{B}}^{\otimes r_{\mathtt{B}} }\otimes \ket{0}_{\mathtt{R}}^{\otimes r_{\mathtt{R}}}.
    \end{align*}
    Given stabilizer generators for $\ket{\Omega}_{\mathtt{ABR}}$, the integers above and stabilizer tableaus for $U_{\mathtt{A}},U_{\mathtt{B}},U_{\mathtt{R}}$ can be computed in time $\poly(q)$.
\end{theorem}

\footnotetext{The original theorem writes the local factors as $\ket{+}$ states. Applying a Hadamard gate to each qubit belonging to a local factor transforms $\ket{+}$ into $\ket{0}$. These Hadamard gates are local Clifford gates and may therefore be absorbed into $U_{\mathtt{A}},U_{\mathtt{B}},U_{\mathtt{R}}$.}

We can now prove:
\begin{proof}[Proof of \cref{lem:low-nullity-form}]
    The given independent stabilizers $P_1,\ldots,P_{q-\nu}$ of $\ket{\psi}$ have a simultaneous $+1$ eigenspace that is a stabilizer code $\mathcal C$ of dimension $\dim(\mathcal C) = \frac{2^q}{2^{q-\nu}} = 2^\nu.$
    Thus, $\mathcal C$ encodes $\nu$ logical qubits, and $\ket{\psi}_{\mathtt{AB}}\in\mathcal C$.

    By \Cref{thm:clifford-reduction}, there is a Clifford encoding unitary $V$ such that $\cB$ is an orthonormal basis of $\mathcal C$:
    \[
    \cB
    \defeq
    \sett{
    V\cdot\ket{0}^{\otimes(q-\nu)}\ket{x}
    }{
    x\in\power{\nu}
    }.
    \]
    Introduce a reference system $\mathtt{R}$ of $\nu$ qubits and define
    \begin{align*}
        \ket{\Omega}_{\mathtt{ABR}}
        &\defeq
        \frac{1}{\sqrt{2^\nu}}
        \sum_{x\in\power{\nu}}
        \ps{
        V\cdot\ket{0}^{\otimes(q-\nu)}\ket{x}
        }_{\mathtt{AB}}
        \ket{x}_{\mathtt{R}}
        \\
        &=
        (V_{\mathtt{AB}}\otimes I_{\mathtt{R}})
        \ps{
        \ket{0}^{\otimes(q-\nu)}
        \otimes
        \ket{\Phi^+}^{\otimes\nu}
        }.
    \end{align*}
    Thus, $\ket{\Omega}_{\mathtt{ABR}}$ is maximally entangled between the registers $\mathtt{AB}$ (namely, the code $\mathcal C$) and $\mathtt{R}$. It is also a stabilizer state because the state inside the parentheses is a tensor product of stabilizer states and $V$ is Clifford.

    Let $\ket{\psi}_{\mathtt{AB}}\in\mathcal C$, and decompose using the basis $\cB$:
    \[
    \ket{\psi}_{\mathtt{AB}}
    =
    \sum_{x\in\power{\nu}}
    \alpha_x
    V\ps{
    \ket{0}^{\otimes(q-\nu)}\ket{x}.
    }
    \]
    Then, observe that
    \[
    \ket{\psi}_{\mathtt{AB}}
    =
    2^{\nu/2}
    \bigl(I_{\mathtt{AB}}\otimes\bra{\eta}_{\mathtt{R}}\bigr)
    \ket{\Omega}_{\mathtt{ABR}},
    \]
    where
    \[
    \ket{\eta}_{\mathtt{R}}
    \defeq
    \sum_{x\in\power{\nu}}
    \overline{\alpha_x}\ket{x}_{\mathtt{R}}.
    \]
    Now invoke \Cref{thm:tripartite-stabilizer} on $\ket{\Omega}_{\mathtt{ABR}}$. The theorem gives local Clifford unitaries $U_{\mathtt{A}},U_{\mathtt{B}},U_{\mathtt{R}}$ that transform the state into a tensor product of local $\ket{0}$ states, Bell pairs, and GHZ states. We denote the resulting state by
    \[
    \ket{\Omega'}_{\mathtt{ABR}}
    \defeq
    (U_{\mathtt{A}}\otimes U_{\mathtt{B}}\otimes U_{\mathtt{R}})\ket{\Omega}_{\mathtt{ABR}}.
    \]
    Because $\ket{\Omega}_{\mathtt{ABR}}$ is maximally entangled between $\mathtt{AB}$ and $\mathtt{R}$, its reduced state on $\mathtt{R}$ is the maximally mixed state. Therefore, the reduced state of $\ket{\Omega'}_{\mathtt{ABR}}$ is also the maximally mixed state, so altogether
    \begin{align*}
        \tr_{\mathtt{AB}} \ps{\ketbraa{\Omega}_{\mathtt{ABR}}}
        &=
        \frac{I_{\mathtt{R}}}{2^\nu}, \\
        \tr_{\mathtt{AB}}(\ketbraa{\Omega'}_{\mathtt{ABR}} ) &= U_{\mathtt{R}} \frac{I_{\mathtt{R}}}{2^\nu}U_{\mathtt{R}}^\dagger
        =
        \frac{I_{\mathtt{R}}}{2^\nu}.
    \end{align*}
    Hence, no qubit of $\mathtt{R}$ can occur in a local $\ket{0}_{\mathtt{R}}$ factor decomposition of $\ket{\Omega'}$. Thus, every qubit of $\mathtt{R}$ belongs to an $\mathtt{AR}$ Bell pair, a $\mathtt{BR}$ Bell pair, or an $\mathtt{ABR}$ GHZ state. Therefore,
    \[
    \nu
    =
    m_{\mathtt{AR}}+m_{\mathtt{BR}}+m_{\mathtt{ABR}}.
    \]
    An $\mathtt{AR}$ or $\mathtt{BR}$ Bell pair contains one qubit from $\mathtt{AB}$, while an $\mathtt{ABR}$ GHZ state contains two. Thus, all factors involving $\mathtt{R}$ contain the following number of qubits from $\mathtt{AB}$:
    \[
    m_{\mathtt{AR}}+m_{\mathtt{BR}}+2m_{\mathtt{ABR}}
    \leq
    2\bigl(m_{\mathtt{AR}}+m_{\mathtt{BR}}+m_{\mathtt{ABR}}\bigr)
    =
    2\nu.
    \]
    Let $\ket{\Theta}_{\mathtt{ABR}}$ be the tensor product of all factors involving $\mathtt{R}$. After locally reordering the remaining factors,
    \[
    \ket{\Omega'}_{\mathtt{ABR}}
    =
    \ket{0}^{\otimes r}
    \otimes
    \ket{\Phi^+}_{\mathtt{AB}}^{\otimes p}
    \otimes
    \ket{\Theta}_{\mathtt{ABR}}.
    \]
    Define
    \begin{align*}
        \ket{\eta'}_{\mathtt{R}}
        &\defeq
        U_{\mathtt{R}}\ket{\eta}_{\mathtt{R}},
        \\
        \ket{\gamma}
        &\defeq
        2^{\nu/2}
        \bigl(I_{\mathtt{AB}}\otimes\bra{\eta'}_{\mathtt{R}}\bigr)
        \ket{\Theta}_{\mathtt{ABR}}.
    \end{align*}
    Using the expression for $\ket{\psi}_{\mathtt{AB}}$ above gives
    \[
    (U_{\mathtt{A}}\otimes U_{\mathtt{B}})\ket{\psi}_{\mathtt{AB}}
    =
    \ket{0}^{\otimes r}
    \otimes
    \ket{\Phi^+}_{\mathtt{AB}}^{\otimes p}
    \otimes
    \ket{\gamma}.
    \]
    Since $\ket{\Theta}_{\mathtt{ABR}}$ contains at most $2\nu$ qubits from $\mathtt{AB}$, the vector $\ket{\gamma}$ is supported on at most $2\nu$ qubits. Taking $D_{\mathtt{A}}, D_{\mathtt{B}}$ to include $U_{\mathtt{A}},U_{\mathtt{B}}$ and the local reordering of the tensor factors proves the result. All the required transformations are computable in polynomial time by the cited stabilizer-tableau reductions.
\end{proof}

\section{Security against Pauli-Sparse Adversaries}
\subsection{Definitions and Security Results}
Let $\Pauli_q=\set{I,X,Y,Z}^{\otimes q}$ denote the Pauli basis group.
Recall that any unitary $U$ expands uniquely in the Pauli basis as
\[
U = \sum_{P \in \Pauli_q} \wh U(P) P,
\qquad
\wh U(P) \defeq 2^{-q} \tr\ps{UP}.
\]
In this subsection, we focus on adversaries whose unitaries are $(q,s)$-Pauli-sparse:
\begin{definition}\label{def:exact pauli sparse}
    A unitary $U$ is \emph{$(q,s)$-Pauli-sparse} if $U$ acts on $q$ qubits and its Pauli spectrum has support size at most $s$, namely:
    \[
    \abs{\supp(\wh U)} \le s.
    \]
    We say that $U$ is \emph{explicit} if its nonzero Pauli coefficients and names can be computed to any requested $b$ bits of precision in classical time $\poly(q,s,b)$.
\end{definition}
\begin{definition}\label{def:pauli-sparse-strategy}
    Let $\cT,s\colon\N\to\N$. A $\cT$-uniform $f$-routing strategy is an (explicit) $(\cT,s)$-uniform Pauli-sparse strategy if the unitary implementations of its resource-preparation and pre-communication channels $\cI^n,\cU^x,\cV^y$ are (explicit) unitaries with Pauli support size at most $s(n)$ each, for every $n$ and $x,y\in\powern$.

    For explicit strategies, the generators run in time $\cT(n,b)$, where $b$ is the specified number of bits of precision for printing the Pauli coefficients ($b$ is omitted when it is clear from context).
\end{definition}

We emphasize that in the case of explicit strategy, we mean printing the explicit description as in \cref{def:exact pauli sparse}; in the case of non-explicit strategy, we mean printing a circuit implementing the unitary in time at most $\cT(n)$, without supplying its Pauli expansion (the circuit can be viewed as an oracle running in time $O(\cT(n))$).

For the remainder of this section, fix an arbitrary input length $n$ and write $d=d(n)$, $s=s(n)$, and $\cT=\cT(n)$. For simplicity, we suppose all unitaries acts on at most $q=q(n)\leq\cT(n)$ qudits each.

The following claim bounds the cost of preparing a product of marginals of states prepared by Pauli-sparse unitaries.
\begin{claim}\label{lem:pauli-sparse-routing-preparations}
    If $\rho_{\tt{\Qb A}}$ is a reduced state of a $(q,s)$-Pauli-sparse unitary, then $\rho_{\tt{\Qb}}\otimes\rho_{\tt{A}}$ is a reduced state of a $(2q,s^2)$-Pauli-sparse unitary.
\end{claim}
\begin{proof}
    Let $W$ be the given unitary and, abusing notation, write $\ket{\rho}_{\tt{\Qb AE}}=W\ket{0^q}$ for a purification of $\rho_{\tt{\Qb A}}$. The unitary $W\otimes W$ prepares $\ket{\rho}_{\tt{\Qb AE}}\otimes\ket{\rho}_{\tt{\Qb' A' E'}}$, and tracing out all registers except $\tt{\Qb A'}$ gives $\rho_{\tt{\Qb}}\otimes\rho_{\tt{A'}}$, yielding the desired state. Since
    \[
    W\otimes W
    =\sum_{P,Q\in\supp(\wh W)}\wh W(P)\wh W(Q)(P\otimes Q),
    \]
    the unitary acts on $2q$ qubits and has Pauli support size at most $s^2$.
\end{proof}

Our main result approximates the mid-protocol fidelity of both explicit and non-explicit Pauli-sparse routing strategies.
\begin{theorem}\label{thm:pauli sparse routing fidelity}
    Fix an input $(x,y)$ and suppose the $f$-routing strategy is $(\cT,s)$-uniform Pauli-sparse.
    Write $F_0=F(\rhopre_{\ol{\mathtt{Q}}\mathtt{A}},\rhopre_{\ol{\mathtt{Q}}}\otimes\rhopre_{\mathtt{A}})$. Then there exists an algorithm, that given $(x,y)$ and $\delta >0$, outputs $\wt{F_0}$ that approximate $\abs{ \wt{F_0} - F_0} \le \delta$, whose complexity is as follows:
    \begin{enumerate}
        \item\label{itm:pauli-fidelity-explicit} If the strategy is \emph{explicit}, the algorithm runs in deterministic classical time $\poly(\cT,d,s,\log(1/\delta))$.\footnote{Here $\cT=\cT(n,b)$ for $b=O(\log(ds/\delta))$.}

        \item\label{itm:pauli-fidelity-non-explicit} If the strategy is \emph{non-explicit}, the algorithm runs in bounded-error quantum time $\poly(\cT,d,s,1/\delta)$.
    \end{enumerate}
\end{theorem}
The proof for explicit strategies is developed in \cref{sec:explicit unitaries}, and for non-explicit strategies in \cref{sec:qs unitaries unexplicit}.

Plugging \cref{thm:pauli sparse routing fidelity} into \cref{thm:universal-fidelity-gap}, we obtain security against Pauli-sparse adversaries.
\begin{corollary}\label{cor:sparse adversary security}
    There are universal constants $c_0$ and $\eps_0>0$ such that the following holds. Let $f=\set{f_n\colon\powern\times\powern\to\power{}}$ be a family of functions.

    Suppose there exists a $(\cT,s)$-uniform Pauli-sparse $f$-routing strategy that is $(1-\eps)$-correct on every input at every input length for some $0\leq\eps\leq\eps_0$.
    \begin{enumerate}
        \item If the strategy is \emph{explicit}, then $f\in\DTIME\ps{(\cT ds)^{c_0}}$.

        \item If the strategy is \emph{non-explicit}, then $f\in\BQTIME\ps{(\cT ds)^{c_0}}$.
    \end{enumerate}
\end{corollary}

\subsection{Explicit Unitaries}\label{sec:explicit unitaries}
We show how to compute the fidelity between two states that are preparable by $(q,s)$ unitaries. In this subsection we assume the unitaries are given explicitly; in \cref{sec:qs unitaries unexplicit} we extend the result to non-explicit unitaries by learning their Pauli coefficients.

To this end, we make use of the following formalism of Uhlmann's theorem \cite{Uhlmann1976}.
\begin{fact}[{\cite[Theorem~3.22]{Watrous18TQI}}]\label{fact from factorizations}
    Let $\rho=XX^\dagger$ and $\sigma=YY^\dagger$ be density matrices on the same space. Then
    \[
    F(\rho,\sigma)=\norm{X^\dagger Y}_1.
    \]
\end{fact}

Our main theorem is then,
\begin{lemma}\label{thm:sparse states to computing fidelity}
    Let $U,V$ be two explicit $(q,s)$-unitaries, and let $\rho, \sigma$ be two reduced states as follows
    \[
    \rho \defeqq \tr_{\bar{\mathtt{A}}}\ps{U\ketbraa{0^q}U^\dagger},
    \qquad
    \sigma \defeqq \tr_{\bar{\mathtt{A}}}\ps{V\ketbraa{0^q}V^\dagger}.
    \]
    For any $\delta> 0$, the fidelity $F(\rho,\sigma)$ can be computed up to additive error $\delta$ by a deterministic classical algorithm in time $\poly(q,s,\log(1/\delta)).$
\end{lemma}

\begin{proof}
    By \cref{thm:facoting nicely midprotocol reduced states X}, there exist two matrices $X,Y$, computable to $O(\log(s/\delta))$ bits of precision in deterministic time $\poly(s,q,\log(1/\delta))$ such that we can decompose the target states
    \begin{align*}
        \rho = XX^\dagger \qquad& \text{where} \qquad X=\sum_{w\in T_X}\ket{\wt\psi_w}\bra w, \\
        \sigma = YY^\dagger \qquad&{where}\qquad Y=\sum_{v\in T_Y}\ket{\wt\phi_v}\bra v.
    \end{align*}
    where $\abs{T_X},\abs{T_Y}\le s.$

    By \cref{fact from factorizations}, the fidelity is $F(\rho,\sigma)=\norm{X^\dagger Y}_1.$
    We now observe that the product matrix has small dimension:
    \[
    M\defeq X^\dagger Y
    =
    \sum_{w\in T_X,\;v\in T_Y}
    \braket{\wt\psi_w}{\wt\phi_v} \cdot \ketbra{w}{v}
    \in
    \C^{\abs{T_X}\times\abs{T_Y}}.
    \]
    Let $\wt M$ be obtained by rounding each entry of $M$ to additive error at most $\delta/(2s^2)$. Since each entry of $M$ is a sum of at most $s^2$ products of bounded Pauli coefficients, $O(\log(s/\delta))$ bits of precision suffice to compute and represent each entry of $\wt M$ with this error bound. As $M$ has at most $s$ rows and columns,
    \[
    \abs{\norm{\wt M}_1-\norm{M}_1}
    \leq \norm{\wt M-M}_1
    \leq \sqrt{s}\,\norm{\wt M-M}_F
    \leq \sqrt{s}\sqrt{s^2}\frac{\delta}{2s^2}
    \leq \frac{\delta}{2}.
    \]
    The first inequality is the reverse triangle inequality $\abs{\norm{A}_1-\norm{B}_1}\leq\norm{A-B}_1$; the second uses
    \(
    \norm{A}_1\leq\sqrt{\rank(A)}\,\norm{A}_F
    \)
    together with $\rank(\wt M-M)\leq s$; for the third, the definition of the Frobenius norm gives
    \[
    \norm{\wt M-M}_F^2
    =\sum_{w\in T_X,\;v\in T_Y}\abs{(\wt M-M)_{wv}}^2
    \leq \abs{T_X}\abs{T_Y}\left(\frac{\delta}{2s^2}\right)^2
    \leq s^2\left(\frac{\delta}{2s^2}\right)^2.
    \]
    Altogether, applying \cref{thm:svd} to $\wt M$ with additive error $\delta/2$ therefore approximates $F(\rho,\sigma)$ to error at most $\delta$, in total time $\poly(q,s,\log(1/\delta))$.
\end{proof}

Altogether, we can compute the mid-protocol fidelity efficiently, assuming the adversaries are $(q,s)$ unitaries.

\begin{proof}[Proof of \cref{thm:pauli sparse routing fidelity}, \cref{itm:pauli-fidelity-explicit}]
    Let $W^{xy}$ be the unitary obtained by composing the resource-preparation unitary, Alice's and Bob's pre-communication unitaries, and a fixed unitary preparing the EPR pair.The EPR-preparation unitary can be chosen explicitly on $2\log_2 d$ qubits, so its Pauli support has size at most $O(d^4)$. The mid-protocol state has a purification of the form
    \[
    \ket{\psi^{xy}}_{\ol{\mathtt{Q}}\mathtt{ABE}}
    =
    W^{xy}\ket{0}^{\otimes O(q)}.
    \]
    Since Pauli support sizes multiply under composition and tensor product, $W^{xy}$ is an explicit $(O(q), O(d^4s^3))$-Pauli-sparse unitary.
    Hence $\rhopre_{\ol{\mathtt{Q}}\mathtt{A}}$ is a reduced state of a $(O(q), O(d^4s^3))$-unitary. By \cref{lem:pauli-sparse-routing-preparations}, $\rhopre_{\ol{\mathtt{Q}}}\otimes\rhopre_{\mathtt{A}}$ is a reduced state of the explicit $(O(q), O(d^8s^6))$-Pauli-sparse unitary $W^{xy}\otimes W^{xy}$. The result follows directly from \cref{thm:sparse states to computing fidelity}.
\end{proof}

\subsection{Deferred Technical Lemmas}

\begin{lemma}
    \label{thm:facoting nicely midprotocol reduced states X}
    Let $U$ be an explicit $(q,s)$-unitary given by its Pauli data, and let $\ket{\psi}=U\ket{0^q}.$
    Then, for every subsystem $\mathtt{S}\subseteq [q]$, there exist $r\leq s$ and a matrix $X_{\mathtt{S}}\in \C^{2^{\abs{\mathtt{S}}} \times r}$ such that
    \[
    \rho_{\mathtt{S}} \defeq \tr_{\bar{\mathtt{S}}}\ketbraa{\psi}
    = X_{\mathtt{S}} X_{\mathtt{S}}^\dagger,
    \]
    Moreover, $X_{\mathtt{S}}$ can be computed to any requested $b$ bits of precision deterministically in time $\poly(q,s,b)$.
\end{lemma}

\begin{proof}
    The key observation is that $\ket\psi$ is a sum of at most $s$ product vectors across $\mathtt{S}:\bar{\mathtt{S}}$. Indeed, each Pauli term maps $\ket{0^q}$ to a computational-basis vector up to a phase, so the Pauli data give
    \[
    \ket\psi=\sum_{j=1}^{m}\alpha_j
    \ket{a_j}_{\mathtt{S}}\ket{b_j}_{\bar{\mathtt{S}}},
    \qquad m\leq s,
    \]
    where the phases are absorbed into $\alpha_j$. Thus the Schmidt rank $R$ of $\ket\psi$ across $\mathtt{S}:\bar{\mathtt{S}}$ is at most $m\leq s$. By the Schmidt decomposition,
    \[
    \rank(\rho_{\mathtt{S}})
    =\rank\ps{\tr_{\bar{\mathtt{S}}}\ketbraa{\psi}}
    =R\leq s.
    \]
    To obtain an efficiently computable factorization, let $b'_1,\ldots,b'_{m'}$ be the distinct strings among $b_1,\ldots,b_m$, and define
    \[
    X_{\mathtt{S}}
    \defeq\sum_{\ell=1}^{m'}
    \ps{\sum_{j:\,b_j=b'_\ell}\alpha_j\ket{a_j}}\bra\ell.
    \]
    Since the vectors $\ket{b'_\ell}$ are orthonormal, tracing out $\bar{\mathtt{S}}$ gives $\rho_{\mathtt{S}}=X_{\mathtt{S}}X_{\mathtt{S}}^\dagger$. Moreover, $X_{\mathtt{S}}$ has $m'\leq m\leq s$ columns.

    Finally, observe that computing all $\alpha_j, a_j,b_j,b_j'$ directly from the explicit Pauli data to $b$ bits of precision is computable in deterministic time $\poly(q,s,b)$. Hence the sparse representation can be computed to $b$ bits of precision in deterministic time $\poly(q,s,b)$.
\end{proof}

\subsection{Non-explicit Unitaries}\label{sec:qs unitaries unexplicit}
We now show how to extend \cref{thm:sparse states to computing fidelity} when the unitaries are given non explicitly. The idea is to learn the unitaries to sufficiently good approximation, and then invoke our earlier result. The main ingredient is therefore an efficient quantum algorithm to learn sparse unitaries:
\begin{theorem}[{\cite[Result~1]{HonjaniHeidari2026}}]\label{thm:learning Pauli coefficients}
    There is an efficient quantum algorithm that, given $\theta>0$ and query access to an unknown unitary $U$, estimates, with constant success probability, all Pauli coefficients of $U$ with magnitude at least $\theta$, up to a global phase factor, with additive error $O(\eps/\theta)$ using $\wt O(\frac{\log(1/\theta)}{\eps^4})$ queries.
\end{theorem}
Using the above we can now generalize our result:
\begin{theorem}\label{thm:learned sparse states to fidelity}
    Let $U,V$ be non-explicit $(q,s)$-Pauli-sparse unitaries given by quantum circuits of description length at most $L$, and let
    \[
    \rho \defeqq \tr_{\bar{\mathtt{A}}}\ps{U\ketbraa{0^q}U^\dagger},
    \qquad
    \sigma \defeqq \tr_{\bar{\mathtt{A}}}\ps{V\ketbraa{0^q}V^\dagger}.
    \]
    For every $\delta\in(0,1)$, there is a bounded-error quantum algorithm that estimates $F(\rho,\sigma)$ up to additive error $\delta$ in time $\poly(q,s,L,1/\delta)$.
\end{theorem}

\begin{proof}
    We first learn the Pauli coefficients of $U$ and show that they give a state close to
    \[
    \ket u \defeq U\ket{0^q}.
    \]
    A symmetric argument would apply to $V$ as well. We then compare the corresponding reduced states.
    Since $U$ is $(q,s)$ Pauli sparse, there exist a set $\cS_U \subseteq \Pauli_q$ of cardinality $\abs{\cS_U} \le s$ such that
    \[
    U=\sum_{P\in \cS_U}\alpha_P P.
    \]
    Let $c_\eta, c_\eps>0$ be sufficiently small universal constants, and fix the following parameters:
    \[
    \eta \defeqq ( c_\eta\delta)^2,
    \quad
    \theta \defeqq \frac{\eta}{2s},
    \quad
    \eps \defeqq c_\eps \theta^2.
    \]
    Here $c_\eps >0$ is a sufficiently small universal constant so that the learner's additive error from \cref{thm:learning Pauli coefficients} is at most $O(\eps/\theta)=O(c_\eps \theta) < \theta/4$.
    Now apply \cref{thm:learning Pauli coefficients} to $U$ with parameter $\theta$, amplifying its success probability to at least $11/12$, to obtain a list of Pauli coefficients $\wt \cS_U \subseteq \Pauli_q$, and their estimated amplitutdes, denoted by $\wt \alpha_P$ for every $P \in \wt \cS_U$. Discard all coefficients of magnitude below $\theta/2$, and write
    \[
    \wt U \defeq \sum_{P\in\wt \cS_U}\wt\alpha_P P.
    \]
    Assume the algorithm succeeds in estimating every coefficient it returns to additive error at most $\theta/4$, up to a common global phase $\phi_U$. This happens with probability at least $\ge 11/12$. Hence every $P\notin\cS_U$ has true coefficient $\alpha_P=0 < \theta/4$, so its estimate $\wt \alpha_P$ cannot survive the cutoff $\theta/2$. Hence the retained set satisfies $\wt\cS_U\subseteq\cS_U$, and therefore $\abs{\wt\cS_U}\leq s$.
    Thus,
    \begin{align*}
        \forall P\in \wt \cS_U\col
        \qquad &\abs{\alpha_P-e^{i\phi_U} \wt\alpha_P}\leq\frac{\theta}{4}, \\
        \forall P\in \cS_U\setminus\wt \cS_U\col
        \qquad &\abs{\alpha_P}<\theta.
    \end{align*}
    Since every Pauli maps $\ket{0^q}$ to a unit vector,
    \begin{align}\label{eq:wt U error}
        \norm{\ket{u}-e^{i\phi_U}\wt U\ket{0^q}}_2
        \leq
        \sum_{P\in\wt \cS_U}\abs{\alpha_P-e^{i\phi_U}\wt\alpha_P}
        +\sum_{P\in \cS_U\setminus\wt \cS_U}\abs{\alpha_P}
        \leq s\frac{\theta}{4} + s\theta
        <\eta.
    \end{align}
    Applying the same argument to $V$, by symmetry, and writing $\ket v\defeq V\ket{0^q}$, with probability at least $5/6$ we obtain an explicit $s$-sparse Pauli sum $\wt V$ and a global phase $\phi_V$ such that
    \[
    \norm{\ket{v}-e^{i\phi_V}\wt V\ket{0^q}}_2<\eta.
    \]
    The learned Pauli sums need not be unitary. We therefore normalize their outputs and consider the corresponding reduced states:
    \begin{gather*}
        \ket{\wt u}\defeqq\frac{\wt U\ket{0^q}}{\norm{\wt U\ket{0^q}}_2},
        \qquad
        \wt\rho\defeqq\tr_{\bar{\mathtt{A}}}\ketbraa{\wt u},
        \\
        \ket{\wt v}\defeqq\frac{\wt V\ket{0^q}}{\norm{\wt V\ket{0^q}}_2},
        \qquad
        \wt\sigma\defeqq\tr_{\bar{\mathtt{A}}}\ketbraa{\wt v}.
    \end{gather*}
    Therefore,
    \begin{align*}
        \norm{\ket u-e^{i\phi_U}\ket{\wt u}}_2
        &\leq
        \norm{\ket u-e^{i\phi_U}\wt U\ket{0^q}}_2
        +\norm{e^{i\phi_U}\ps{\wt U\ket{0^q}-\ket{\wt u}}}_2 \\
        &= \norm{\ket u-e^{i\phi_U}\wt U\ket{0^q}}_2
        +\norm{\wt U\ket{0^q}-\ket{\wt u}}_2\\
        &=\norm{\ket u-e^{i\phi_U}\wt U\ket{0^q}}_2
        +\norm{\ps{\norm{\wt U\ket{0^q}}_2-1}\ket{\wt u}}_2 \\
        &=
        \norm{\ket u-e^{i\phi_U}\wt U\ket{0^q}}_2
        +\abs{\norm{\wt U\ket{0^q}}_2-1} \cdot 1 \\
        &\leq
        2\norm{\ket u-e^{i\phi_U}\wt U\ket{0^q}}_2 \\
        &<2\eta,
    \end{align*}
    where the ultimate inequality follows from \cref{eq:wt U error}, and the penultimate inequality follows from the reverse triangle inequality which states that $\abs{\norm{x}_2-\norm{y}_2}\leq\norm{x-y}_2$ for any two vectors $x,y$):
    \begin{align*}
        \abs{1 - \norm{\wt U \ket{0^q} }_2 } &=
        \abs{\norm{\ket u}_2-\norm{\wt U\ket{0^q}}_2} \\
        &= \abs{\norm{\ket u}_2-\norm{e^{i\phi_U}\wt U\ket{0^q}}_2} \\
        &\leq \norm{\ket u-e^{i\phi_U}\wt U\ket{0^q}}_2.
    \end{align*}
    We now use contractivity under partial trace: partial trace is a quantum channel and cannot increase the trace norm between two states. Together with the trace-distance formula for pure states, this gives
    \begin{align*}
        \norm{\rho-\wt\rho}_1
        &=
        \norm{\tr_{\bar{\mathtt{A}}}\ps{\ketbraa u-\ketbraa{\wt u}}}_1\\
        &\leq
        \norm{\ketbraa u-\ketbraa{\wt u}}_1\\
        &=2\sqrt{1-\abs{\braket{u}{\wt u}}^2}\\
        &\leq
        2\norm{\ket u-e^{i\phi_U}\ket{\wt u}}_2<4\eta.
    \end{align*}
    By symmetry, $\norm{\sigma-\wt\sigma}_1<4\eta$. Therefore, by \cref{lem:fvg,thm:fidelity triangle}, for sufficiently small constant $c_\eta$,
    \[
    \abs{F(\rho,\sigma)-F(\wt\rho,\wt\sigma)}
    =O(\sqrt\eta) = O(c_\eta \delta) \leq \frac{\delta}{2}.
    \]
    The computation in \cref{thm:sparse states to computing fidelity} uses only the explicit sparse descriptions of the normalized output vectors. It therefore computes a value $\wt F$ such that
    \[
    \abs{\wt F-F(\wt\rho,\wt\sigma)}\leq\frac{\delta}{2}.
    \]
    Thus $\abs{\wt F-F(\rho,\sigma)}\leq\delta$. The total running time is $\poly(q,s,L,1/\delta)$.
\end{proof}

\begin{proof}[Proof of \cref{thm:pauli sparse routing fidelity}, \cref{itm:pauli-fidelity-non-explicit}]
    Construct $W^{xy}$ as in the proof of \cref{itm:pauli-fidelity-explicit}. The unitaries $W^{xy}$ and $W^{xy}\otimes W^{xy}$ prepare the two states defining $F_0$ as reduced states. Similarly to the proof of \cref{itm:pauli-fidelity-explicit}, they act on $O(q)$ qubits, have Pauli support sizes at most $O(d^4s^3)$ and ${O(d^8s^6)}$, respectively, and have circuits of size $O(\cT)+\poly(d)$. Applying \cref{thm:learned sparse states to fidelity} gives the required bounded-error quantum algorithm.
\end{proof}

\section*{AI Disclosure}
This research was conducted by humans, and was assisted by AI tools. All research direction, discovery, lemma statements, hypothesis were originated and tested by the the authors, and ChatGPT's Sol 5.6 assisted to prove them. We repeatedly revised and interacted with ChatGPT to produce the final proofs.

\printbibliography

\end{document}